\documentclass[11pt]{article}

\usepackage{amsmath,amsfonts,amssymb}

\usepackage[numbers,sort&compress]{natbib}
\usepackage{graphicx}

\usepackage{hyperref}
\hypersetup{
   colorlinks   =  true,
    linkcolor    = blue,
    citecolor    = red,
     urlcolor	=magenta,     
}

\usepackage{amsthm}
\theoremstyle{definition} %Normal font

\newtheorem{remark}{Remark}
\theoremstyle{theorem} %Italics

\newtheorem{lemma}{Lemma}
\newtheorem{proposition}{Proposition}
\newtheorem{corollary}{Corollary}

\newcommand{\sect}[1]{\setcounter{equation}{0}\section{#1}}

\newcommand\be{\begin{equation}}
\newcommand\ee{\end{equation}}
\newcommand\bea{\begin{eqnarray}}
\newcommand\eea{\end{eqnarray}}

\newcommand\ka{\Lambda}
 \newcommand\ca{\mathfrak{c}}
 \newcommand\nh{\mathfrak{n}}
\newcommand\w{\wedge}

 \newcommand\ro{\eta }
   \newcommand\nullp{{\rm N}}
 \newcommand\dd{{\rm d}}
 \newcommand\cas{ {\mathcal C} }

\newcommand\fun{ {A} }
\newcommand\funb{ {B} }

\newcommand\func{ {\Theta} }
\newcommand\fund{ {\Upsilon} }

\DeclareMathOperator\spn{span}

\def\L{\Lambda}
 \def\mf {\mathfrak}
 \def\>#1{{\bf #1}} 
\def\s{s}

\begin{document}

\

\begin{center}
\baselineskip 24 pt {\LARGE \bf  
Noncommutative (A)dS and Minkowski spacetimes\\ from  quantum Lorentz subgroups}
\end{center}

   \medskip
 \medskip

\begin{center}

{\sc Angel Ballesteros$^{1}$, Ivan Gutierrez-Sagredo$^{1,2}$ and Francisco J.~Herranz$^{1}$}

\medskip
{$^1$Departamento de F\'isica, Universidad de Burgos, 
09001 Burgos, Spain}

{$^2$Departamento de Matem\'aticas y Computaci\'on, Universidad de Burgos, 
09001 Burgos, Spain}
\medskip
 
e-mail: {\href{mailto:angelb@ubu.es}{angelb@ubu.es}, \href{mailto:igsagredo@ubu.es}{igsagredo@ubu.es}, \href{mailto:fjherranz@ubu.es}{fjherranz@ubu.es}}

\end{center}

\medskip

\begin{abstract}
\noindent

The complete classification of classical $r$-matrices generating quantum deformations of the (3+1)-dimensional (A)dS and Poincar\'e groups such that their Lorentz sector is a quantum subgroup is presented. It is found that there exists three classes of such $r$-matrices, one of them being a novel two-parametric one. The (A)dS and Minkowskian Poisson homogeneous spaces corresponding to these three deformations are explicitly constructed in both local and ambient coordinates. Their quantization is performed, thus giving rise to the associated noncommutative spacetimes, that in the Minkowski case are naturally expressed in terms of quantum null-plane coordinates, and they are always defined by homogeneous quadratic algebras. Finally, non-relativistic and ultra-relativistic limits giving rise to novel Newtonian and Carrollian noncommutative spacetimes are also presented.
 \end{abstract}

 \medskip
\medskip\medskip

\noindent
PACS:   \quad 02.20.Uw \quad  03.30.+p \quad 04.60.-m

 \medskip

\noindent
KEYWORDS:    quantum groups;  Minkowski spacetime; (A)dS spacetimes; cosmological constant;   noncommutative spaces;  Lie bialgebras;  Poisson-Lie groups;    Poisson  homogeneous spaces;  contractions; Galilei; Carroll

\newpage

\tableofcontents

%%%%%%%%%%%%%%%%%%%%%%%%%%%%%%%%%%%%%%%%%%%%%%%%%%%

\sect{Introduction}

It is well-known that quantum Poincar\'e and (anti-)de Sitter (hereafter  (A)dS) groups (see~\cite{LRNT1991, Maslanka1993, MR1994, Zakrzewski1994poincare, Zakrzewski1995, nullplane95, Zakrzewski1997, PW1996, CatherineDD,BHMN2017kappa3+1, PoinDD} and references therein) provide Hopf algebra deformations of relativistic symmetries which can be used to construct Deformed Special Relativity (DSR) models~\cite{Amelino-Camelia2001testable,Kowalski-Glikman2001,Amelino-Camelia2002planckian,MS2002,LN2003versus,Amelino-Camelia2010symmetry} in which the quantum deformation parameter is assumed to be related to the Planck scale. In this setting, quantum gravity effects are assumed to be described in a schematic way through the associated noncommutative Minkowski and (A)dS spacetimes. Moreover, in the (A)dS case the interplay between the non-vanishing cosmological constant $\Lambda$ and the Planck scale parameter can be explicitly described.

In this approach, the natural question concerning the complete classification of quantum Poincar\'e and (A)dS groups in (3+1) dimensions arises, since solving this problem will provide the full set of available Hopf algebra deformations of relativistic symmetries. In the Poincar\'e case, such classification was given in~\cite{PW1996} at the quantum group level, while classical $r$-matrices generating all possible quantum deformations were classified in~\cite{Zakrzewski1997}. However, to the best of our knowledge, similar results are not available for the (A)dS case yet. Besides its intrinsic mathematical relevance, the latter result would be indeed relevant for DSR Hopf algebra models dealing with the propagation of particles on quantum spacetimes at cosmological distances.

The aim of this paper is to contribute to  fill this gap, at least partially, by  providing the complete classification of classical $r$-matrices of the Poincar\'e and (A)dS groups endowed with an outstanding condition: that their $\mathfrak{so}(3,1)$ Lorentz sector has to be a Hopf subalgebra after the full quantum group is constructed. As we will show, this problem can be fully and explicitly solved by considering that  the complete classification of classical $r$-matrices of the $\mathfrak{so}(3,1)$ Lorentz group is known~\cite{Zakrzewski1994lorentz}, and the answer leads to only three families of $r$-matrices generating quantum Lorentz subgroups. One of them is a novel biparametric $r$-matrix  whose associated noncommutative spacetime is constructed and analysed in detail. Remarkably enough, the cosmological constant parameter $\Lambda$ does not appear in any of the three $r$-matrices, which means that these three solutions are invariant under the flat limit $\Lambda\to 0$. Moreover, it is also explicitly shown that there does not exist any non-trivial quantum Poincar\'e or (A)dS group for which the Lorentz sector is a non-deformed one.

It is worth stressing that  the well-known $\kappa$-deformations~\cite{LRNT1991,Maslanka1993,MR1994,Zakrzewski1994poincare,ASS2004,BHMN2017kappa3+1,BGH2019kappaAdS3+1} do not belong to the class of quantum groups here studied. Therefore, the new noncommutative Minkowski and (A)dS spacetimes here presented would be worth to be considered in order to construct new DSR models. As we will see in detail, in the Minkowskian case these novel noncommutative spacetimes are defined through a  homogeneous quadratic algebra of spacetime operators, and the natural setting for them is to consider noncommutative null-plane coordinates. Consequently, these new noncommutative Minkowskian spacetimes are quite different from the   $\kappa$-Minkowski  noncommutative spacetime~\cite{Maslanka1993} and from its linear-algebraic generalizations~\cite{Daszkiewicz2008,BorowiecPachol2009jordanian,BP2014extendedkappa}  (see also references therein).

The paper is organized as follows. 
In the next Section the  family of (3+1)-dimensional (A)dS and Poincar\'e Lie algebras, Lie  groups and classical homogeneous  spacetimes is presented
in a unified setting by making use of   the cosmological constant $\L$, where these three Lie groups are collectively denoted by   $G_\L$. After recalling in Section~\ref{s3} the basic tools and concepts needed for the paper, we present our main results in Section~\ref{s4} by taking into account that any quantum deformation   
 for $G_\L$ has to be generated by an underlying coboundary Lie bialgebra structure coming from a    solution of the modified classical Yang--Baxter equation. In particular, we firstly prove that there does not exist any quantum deformation for the family $G_\L$ with an undeformed Lorentz subgroup. Secondly,  we explicitly compute all  classical $r$-matrices for the Poincar\'e and (A)dS groups that present a Lorentz sub-Lie bialgebra structure.      It comes out that there only exist three types of such $r$-matrices, which are explicitly obtained: two of them depends on a single deformation parameter, while the other one is a two-parametric deformation. The (2+1)-dimensional counterpart of this classification is also presented in 
 Section~\ref{s41} and  the differences  that   arise between both dimensions are analysed.
     Furthermore,  our results entail a classification of (3+1)-dimensional (A)dS and Minkowski noncommutative spacetimes associated to deformations with a quantum Lorentz subgroup, whose semiclassical counterpart corresponds to Poisson homogeneous spaces  of Poisson Lorentz subgroup type.  The latter are deduced in Section~\ref{s5}  for the three types of quantum deformations by considering   both local and ambient coordinates, and are summarized in Table~\ref{table1}.      
     In addition, we   perform  in Section~\ref{s6} the  non-relativistic and ultra-relativistic limits of  the above classification, giving rise to  
   Newtonian and Carrollian  quantum deformations, for which the relevant role formerly played by the Lorentz subgroup is now replaced by the three-dimensional (3D) Euclidean one generated by rotations and boost transformations. Their corresponding noncommutative spacetimes are also derived and presented in Table~\ref{table2}.
  Finally some remarks and open problems close the paper.
  
   \newpage

%%%%%%%%%%%%%%%%%%%%%%%%%%%%%%%%%%%%%%%%%%%%%%%%%%%

\sect{(A)dS and Poincar\'e  groups: a joint description}

In our framework we will consider the (3+1)D Poincar\'e and (A)dS Lie algebras expressed in a unified setting as  a one-parametric family of Lie algebras  denoted by $\mathfrak g_\L$    depending   explicitly on the cosmological constant parameter $\Lambda$. In the usual kinematical basis, spanned by the   generators of time translations $P_0$, spatial translations  $\>P=(P_1,P_2,P_3)$, boost transformations $\>K=(K_1,K_2,K_3)$  and rotations $\>J=(J_1,J_2,J_3)$, the      commutation relations of  $\mathfrak g_\L$ read
\begin{equation}
\label{eq:ads_3+1}
\begin{array}{lll}
[J_a,J_b]=\epsilon_{abc}J_c ,& \quad [J_a,P_b]=\epsilon_{abc}P_c , &\quad [J_a,K_b]=\epsilon_{abc}K_c , \\[2pt]
\displaystyle{
  [K_a,P_0]=P_a  } , &\quad\displaystyle{[K_a,P_b]=\delta_{ab} P_0} ,    &\quad
  \displaystyle{[K_a,K_b]=-\epsilon_{abc} J_c} , 
\\[2pt][P_0,P_a]=- \L  K_a , &\quad   [P_a,P_b]=\L \epsilon_{abc}J_c , &\quad[J_a,P_0]=0    .
\end{array}
\end{equation}
From now on sum over repeated indices will be understood  unless otherwise stated. Hereafter, Latin indices run as $a,b,c=1,2,3$, and Greek ones as  $\mu=0,1,2,3$. The family of Lie algebras  $\mathfrak g_\L$ encompasses the dS   algebra $\mathfrak{so}(4,1)$ when $\Lambda > 0$, the AdS   algebra $\mathfrak{so}(3,2)$ if $\Lambda < 0$ and the Poincar\'e one  $\mathfrak{iso}(3,1)$  for $\Lambda = 0$. 
As a vector space, $\mathfrak g_\L$ can be decomposed in the form
\be
\mathfrak g_\L = \mathfrak h \oplus \mathfrak t ,  \qquad \mathfrak h = \spn{ \{\mathbf{K}, \mathbf{J} \}}=\mathfrak{so}(3,1),\qquad \mathfrak{t} = \spn{ \{P_0, \mathbf{P} \}} ,
\label{decom}
\ee
where    $\mathfrak h$ is the Lorentz subalgebra  and $\mathfrak{t} $ is the translation sector.

A faithful representation of   $\mathfrak g_\L$ (\ref{eq:ads_3+1}), $\rho : \mathfrak g_\L \rightarrow \text{End}(\mathbb R ^5)$, for a generic   element $X\in \mathfrak g_\L$ is given by~\cite{BGH2019kappaAdS3+1}
\begin{equation}
\label{eq:repg}
\rho(X)=  x^\mu \rho(P_\mu)   +  \xi^a \rho(K_a) +  \theta^a \rho(J_a) =\left(\begin{array}{ccccc}
0&\L x^0&-\L x^1&-\L x^2&-\L x^3\cr 
x^0 &0&\xi^1&\xi^2&\xi^3\cr 
x^1 &\xi^1&0&-\theta^3&\theta^2\cr 
x^2 &\xi^2&\theta^3&0&-\theta^1\cr 
x^3 &\xi^3&-\theta^2&\theta^1&0
\end{array}\right) .
\end{equation}
 The corresponding exponentiation provides   a one-parametric family of Lie groups,   denoted by $G_\Lambda$, with Lie algebra  $\mathfrak g_\L$. Hence $G_\Lambda$ contains  the dS   group $\mathrm{SO}(4,1)$ for $\Lambda > 0$, the  AdS group $\mathrm{SO}(3,2)$  for $\Lambda < 0$, and  the Poincar\'e  one $\mathrm{ISO}(3,1)$ for $\Lambda = 0$. Note that the exponentiation of (\ref{eq:repg}) only recovers the connected component to the identity of these three Lie groups, but the complete description of these groups can be obtained by considering the parity  and time-reversal involutive automorphisms.   The family $G_\Lambda$ can be parametrized in terms of  local coordinates $\{x^\mu, \xi^a, \theta^a \}$ in the form
 \begin{align}
\begin{split}
\label{eq:Gm}
&G_\L= \exp{x^0 \rho(P_0)} \exp{x^1 \rho(P_1)} \exp{x^2 \rho(P_2)} \exp{x^3 \rho(P_3)} \\
&\qquad\quad\times \exp{\xi^1 \rho(K_1)} \exp{\xi^2 \rho(K_2)} \exp{\xi^3 \rho(K_3)}
 \exp{\theta^1 \rho(J_1)} \exp{\theta^2 \rho(J_2)} \exp{\theta^3 \rho(J_3)} .
\end{split}
\end{align}
 These coordinates are the so-called exponential coordinates of the second kind on $G_\Lambda$.

Therefore $G_\L$ comprises the isometry Lie groups of the (3+1)D  Minkowski and (A)dS spacetimes, which we   denote collectively  by $M_\L$ and have a constant sectional curvature given by $-\L$. For each of these three spacetimes, the stabilizer of a point is the Lorentz subgroup  $H$ with Lie algebra $\mathfrak h  $ (\ref{decom}). Thus there exists a global isometry between $M_\L$ and the left coset space  $G_\L/H$, so that we write 
\be
M_\L = G_\L/H,\qquad  H={\rm SO}(3,1)= \langle \>K,\>J \rangle.
\label{ml}
\ee
Hence $M_\L$ is a family of symmetric  homogeneous spaces, and we can identify their tangent space at every point $m = g H \in M_\L$,  $g\in G_\L$,  with the translation sector, \emph{i.e.} 
\begin{equation}
T_{m} (M_\L) = T_{gH} (G_\L/H) \simeq \mathfrak{g}_\L/\mathfrak{h} \simeq \mathfrak{t} = \spn{ \{P_0, \mathbf{P} \}} .
\end{equation}
The four Lie group local coordinates  $x^\mu$ in    \eqref{eq:Gm},  associated to the translation generators $P_\mu$, descend to coordinates on $M_\L$. 

The representation of  $G_\L$ (\ref{eq:Gm}), coming   from (\ref{eq:repg}), allows us to  consider $G_\L$ as the isometry group of the 5D linear space $(\mathbb R^5, \mathbf I_{\ka})$, with canonical linear ambient coordinates
$ (\s^4,\s^\mu)$, such that $\mathbf I_{\ka}$ is the bilinear form given by
\begin{align}
&\mathbf I_{\ka}={\rm diag}(+1,-\ka,\ka,\ka ,\ka),
\label{bf}
\end{align}
fulfilling that $ G_\ka^T \, \mathbf I_{\ka}\, G_\ka =\mathbf I_{\ka} $.  The point with ambient coordinates $O =(1,0,0,0,0)$  is invariant under the action of the Lorentz subgroup $H$, and will be taken as the origin of $M_\L$. The orbit
passing through $O$ corresponds to the (3+1)D     spacetime  $M_\L$ (\ref{ml}) defined by the pseudosphere
\begin{equation}
\Sigma_\ka\equiv ( s^4)^2 - \ka  (\s^0)^2 +\ka \bigl( (\s^1)^2+ (\s^2)^2+ (\s^3)^2 \bigr)=1  ,
\label{pseudo}
\end{equation}
determined by $\mathbf I_{\ka}$ (\ref{bf}).  In the flat limit   $\ka\to 0$, the Minkowski spacetime will be identified with the hyperplane $\s^4=+1$ containing $O$. We remark that the coordinates
\be
q^\mu=\frac{\s^\mu}{\s^4}
\label{beltrami}
\ee
are just the  Beltrami projective coordinates in  $M_\L$ which can be obtained through
 the  projection with pole  
$(0,0,0,0,0)\in \mathbb R^{5}$ of a point with ambient coordinates $(\s^4,\s^\mu) $ onto the projective hyperplane with $\s^4=+1$ (see~\cite{BGH2020snyder} for details).  

Now the set of spacetime local coordinates $x^\mu$ can be introduced through the following action onto the origin $O$ of the one-parameter subgroups of $G_\L$ (\ref{eq:Gm})~\cite{BGH2019kappaAdS3+1}
\be
(s^4,s^\mu)^T= \exp{x^0 \rho(P_0)} \exp{x^1 \rho(P_1)} \exp{x^2 \rho(P_2)} \exp{x^3 \rho(P_3)} O^T , 
\ee
thus yielding
\begin{align} 
\begin{split}
\label{ambientspacecoords}
&s^4=\cos( \ro x^0) \cosh( \ro x^1) \cosh( \ro x^2 )\cosh( \ro x^3 ), \\
&s^0=\frac {\sin( \ro x^0)}\ro  \cosh( \ro x^1 )\cosh( \ro x^2 )\cosh( \ro x^3), \\
&s^1=\frac {\sinh( \ro x^1) }\ro   \cosh (\ro x^2) \cosh( \ro x^3), \\
&s^2=\frac { \sinh (\ro x^2)} \ro\cosh( \ro x^3), \\
&s^3=\frac { \sinh (\ro x^3)} \ro ,
\end{split}
\end{align}
where the parameter  $\ro$ is defined by 
\be
\ro^2:=-\Lambda .
\label{constant}
\ee
Thus $\ro$ is real for the AdS space
and a pure imaginary number  for the  dS one. 
The four spacetime coordinates  $x^\mu$  are called geodesic parallel coordinates~\cite{BHMN2014sigma}, and can be regarded as a generalization of the flat  Cartesian coordinates to non-zero curvature. In fact, 
under  the vanishing cosmological constant limit, $\ro\to 0$,    the ambient \eqref{ambientspacecoords} and Beltrami (\ref{beltrami}) coordinates reduce to the usual Cartesian ones in the Minkowski spacetime: 
\be
(\s^4,\s^\mu )\equiv (1,q^\mu )\equiv  (1,x^\mu ).
\label{bzz}
\ee
In ambient coordinates, the  time-like metric  on the   (3+1)D   $M_\L$ spacetime comes from the pseudosphere (\ref{pseudo}) and
 turns out to be~\cite{BGH2020snyder}
\bea
&&\!\!\!\!\!\!\!\!\!\!\!\! {\rm d} \sigma_\ka^2=\left.\frac {1}{-\ka}
\biggl( ({\rm d} \s^4)^2-\ka \bigl( ( \dd \s^0)^2 - (\dd \s^1)^2- (\dd\s^2)^2- (\dd\s^3)^2 \bigr) 
\biggr)\right|_{\Sigma_\ka} \!\!  \nonumber\\[4pt]
&&=    \frac{  -\ka\left(  \s^0{\rm d} \s^0- \s^1 \dd\s^1-\s^2 \dd\s^2-\s^3 \dd\s^3  \right)^2} {1 +  \ka  \bigl( (\s^0)^2 - (\s^1)^2- (\s^2)^2- (\s^3)^2 \bigr)    }+  (\dd \s^0)^2- (\dd \s^1)^2  - ( \dd \s^2)^2- (\dd \s^3)^2  .
\label{metric1}
\eea
From this expression the metric in terms of Beltrami coordinates (\ref{beltrami}) can be obtained~\cite{BGH2020snyder}, and in    geodesic parallel coordinates  \eqref{ambientspacecoords} the metric reads~\cite{BGH2019kappaAdS3+1}
\bea
&&\dd\sigma_\ka^2 =\cosh^2(\ro x^1) \cosh^2(\ro x^2)\cosh^2(\ro x^3) (\dd x^0)^2-\cosh^2(\ro
x^2) \cosh^2(\ro x^3)(\dd x^1)^2 \nonumber\\ 
&&\qquad \quad    -\cosh^2(\ro x^3)( \dd x^2)^2- (\dd x^3)^2 \, .
\label{metric2}
\eea
Indeed, when $\ro\to 0$ expressions (\ref{metric1}) and  (\ref{metric2})  reduce to the usual metric on the Minkowski spacetime:
\be
\dd\sigma_0^2 =  (\dd \s^0)^2- (\dd \s^1)^2  - ( \dd \s^2)^2- (\dd \s^3)^2 =  (\dd x^0)^2- (\dd x^1)^2  - ( \dd x^2)^2- (\dd x^3)^2.
\ee

%%%%%%%%%%%%%%%%%%%%%%%%%%%%%%%%%%%%%%%%%%%%%%%%%%%

\sect{Lie bialgebras and Poisson homogeneous spaces}
\label{s3}

 Let us briefly recall the basic notions needed for the paper, and set up the notation; for more details we refer to~\cite{ChariPressley1994}. A Lie bialgebra is a pair $(\mf g, \delta)$ where $\mf g$ is a Lie algebra and 
\begin{equation}
\delta : \mf g \rightarrow   \mf g\wedge \mf g
\end{equation}
is a linear map called cocommutator and satisfying the following conditions 
\begin{equation}
\begin{split}
&\sum_{cycl} (\delta \otimes \mathrm{id}) \circ \delta (X) = 0, \qquad     \forall X \in \mf g, \\
&\delta ([X,Y]) = \mathrm{ad}_X \delta (Y) - \mathrm{ad}_Y \delta (X), \qquad     \forall X, Y \in \mf g. \\
\end{split}
\end{equation}
The first expression is called co-Jacobi condition since this is equivalent to requiring that the transpose map $^t \delta : \mf g^* \w \mf g^* \rightarrow \mf g^*$ satisfies the Jacobi identity. The second relation states that the map $\delta$ is a 1-cocycle in the Chevalley--Eilenberg cohomology with values in $\mf g \wedge \mf g$. Given a Lie bialgebra $(\mathfrak{g},\delta)$ and a  Lie subalgebra  $\mathfrak{h}$ of $\mathfrak{g}$,  the pair  $(\mathfrak{h},\delta |_\mathfrak{h})$ is said to be  a sub-Lie bialgebra of $(\mathfrak{g},\delta)$ in the case that  $\delta(\mathfrak{h}) \subset \mathfrak{h} \wedge \mathfrak{h}$.

For some Lie bialgebra structures the cocommutator map can be completely  defined in terms of an element $r \in  \mf g\wedge \mf g$. In these cases, called 1-coboundaries, the map
\begin{equation}
\label{eq:deltar}
\delta_r(X) = \mathrm{ad}_X (r)
\end{equation}
defines a Lie bialgebra structure  if and only if $r$ fulfils the modified classical Yang--Baxter equation (mCYBE)
\begin{equation}
\label{mCYBE}
\textrm{ad}_X [[r,r]] = 0,\qquad     \forall X \in \mf g ,
\end{equation}
 where the algebraic Schouten bracket  is defined by 
\begin{equation}
[[r,r]] := [r_{12},r_{13}] + [r_{12},r_{23}] +[r_{13},r_{23}] .
\end{equation}
Therefore   an element $r \in  \mf g\wedge \mf g$ satisifes  the mCYBE (\ref{mCYBE})  if and only if 
\be
[[r,r]] \in \biggl( \bigwedge^3 \mf g \biggr)_{\mf g} \, ,
\label{mCYBE2}
\ee
\emph{i.e.} if  its algebraic Schouten bracket is a $\mf g$-invariant element of $\bigwedge^3 \mf g$. Particular solutions of the mCYBE are those which fulfil the classical Yang--Baxter equation (CYBE)
\begin{equation}
\label{CYBE}
[[r,r]] = 0,
\end{equation}
that is,  the algebraic Schouten bracket vanishes identically. Solutions of the CYBE are called triangular (or nonstandard)     classical $r$-matrices, while solutions of the mCYBE that are not solutions of the CYBE are called quasitriangular (or standard)   classical $r$-matrices. A Lie bialgebra $(\mf g, \delta)$ is then called a coboundary one  if its cocommutator $\delta$ is of the form \eqref{eq:deltar} for some solution of the mCYBE (\ref{mCYBE}).

A Poisson--Lie (PL) group is the global object integrating a Lie bialgebra structure. More explicitly, a PL group is a pair $(G,\Pi)$ where $G$ is a Lie group and $\Pi$ is a Poisson structure such that the Lie group multiplication $\mu : G \times G \rightarrow G$ is a Poisson map with respect to $\Pi$ on $G$ and  the product Poisson structure $\Pi_{G\times G} = \Pi \oplus \Pi$ on $G \times G$. The relation between the Poisson bivector field and the Poisson bracket is given by 
\begin{equation}
(\mathrm d f_1 \otimes \mathrm d f_2) \Pi = \{f_1,f_2 \}_\Pi \, .
\end{equation}
A Lie subgroup $H$ of $G$ is said to be a PL subgroup of $(G,\Pi)$ if $(H,\Pi |_H)$ is a Poisson submanifold  of $(G,\Pi)$.
  A PL group is called coboundary if its tangent Lie bialgebra is a coboundary one. Let $(G,\Pi)$ be a coboundary PL group with tangent Lie bialgebra $(\mf g, \delta_r)$ \eqref{eq:deltar}, then the Poisson bivector on $G$ is given by the Sklyanin bivector  
\begin{equation}
\label{eq:sklyaninbiv}
\Pi = r^{ij} \left( X^L_i \otimes X^L_j - X^R_i \otimes X^R_j \right),
\end{equation}
where $ X^L_i $ and $X^R_i$ denote the left- and right-invariant vector fields, respectively.

  A    Poisson manifold $(M,\pi)$  is  a manifold $M$ endowed  with a  Poisson structure $\pi$ on $M$. A Poisson homogeneous space (PHS) for a PL group $(G,\Pi)$ is a Poisson manifold $(M,\pi)$ which is endowed with a  transitive  group action   $\alpha : (G \times M, \Pi \oplus \pi ) \rightarrow (M, \pi)$  which is a Poisson map.   When the manifold is a  homogeneous space,   $M=G/H$, there does exist a   
 distinguished  class of PHS  for which the Poisson structure $\pi$  on $M$ can be obtained by canonical projection of the PL structure $\Pi$ on $G$. In terms of the underlying Lie bialgebra $(\mf g, \delta)$ of $(G,\Pi)$, this  requirement 
 corresponds  to imposing the so-called    coisotropy condition  for the cocommutator $\delta$ with respect to the isotropy subalgebra $\mathfrak h$ of $H$ that is given by~\cite{Lu1990thesis,Ciccoli2006,BMN2017homogeneous} 
 \be
\delta(\mathfrak h) \subset \mathfrak h \wedge \mathfrak g.
\label{coisotropy}
\ee
A  particular case fulfilling the above condition is obtained when the Lie subalgebra $\mathfrak h $ is also a  sub-Lie bialgebra:  
\be
\delta\left(\mathfrak{h}\right) \subset \mathfrak{h} \wedge \mathfrak{h} ,
\label{eq:subLiebialg}
\ee
which implies that  the PHS  is constructed through an isotropy subgroup $H$  such that   $(H,\Pi |_H)$  is a PL subgroup of  $(G,\Pi)$ and this
 is called a ``PHS of Poisson subgroup type".
 
We recall that since a quantum group    is the quantization   of a PL  group $(G,\Pi)$, the    quantization of a coisotropic PHS  $(M=G/H,\pi)$ fulfilling (\ref{coisotropy}) provides  a quantum homogeneous space (or  noncommutative space)  onto which the quantum group   co-acts covariantly \cite{Dijkhuizen1994}. The  coisotropy condition (\ref{coisotropy}) ensures that the commutation relations that define the noncommutative space  at the first-order in all the quantum coordinates close on  a Lie subalgebra which is just the   annihilator $\mathfrak{h}^\perp$ of $\mathfrak{h}$ on the dual Lie algebra $\mathfrak g^*$~\cite{Ciccoli2006,BGM2019coreductive,GH2021symmetry}.

By taking into account the above concepts, we would like to remark that in the next Section
 we shall deduce all the possible Lie bialgebra structures   for the family $\mf g_\L$ (\ref{decom}) such that the Lorentz subalgebra 
 $\mathfrak{h}$  is a sub-Lie bialgebra of   $(\mathfrak{g}_\L,\delta)$. 
We will find three types of such bialgebras being   all of them coboundaries (\ref{eq:deltar}), so coming from solutions  $r \in  \mf g_\L\wedge \mf g_\L$ of the mCYBE (\ref{mCYBE}).  Consequently, each of them provides a PL group 
 $(G_\L,\Pi)$ for the family $G_\L$ (\ref{eq:Gm}) through  the   Sklyanin bivector (\ref{eq:sklyaninbiv}) and, by construction, the Lorentz subgroup $H$ becomes   a PL subgroup $(H,\Pi |_H)$ of $(G_\L,\Pi)$. Furthermore, each of these  PL groups leads to a PHS $(M_\L = G_\L/H,\pi)$
where the Poisson structure $\pi$ on $M_\L$   is obtained by canonical projection of the PL structure $\Pi$ on $G_\L$.
 In Section~\ref{s5} we will first compute explicitly all such PHS and, secondly, we will provide their complete quantum versions, thus obtaining all the 
  (3+1)D noncommutative (A)dS and Minkowski spacetimes which are covariant under the only quantum (A)dS and Poincar\'e groups for which the Lorentz subalgebra has a quantum subgroup structure.

%%%%%%%%%%%%%%%%%%%%%%%%%%%%%%%%%%%%%%%%%%%%%%%%%%%

\sect{(A)dS and Poincar\'e bialgebras with a  Lorentz sub-bialgebra}
\label{s4}

We proceed 
  to obtain all the    PL groups $(G_\L,\Pi)$ associated with Lie bialgebras $(\mathfrak{g}_\L,\delta)$ for which the Lorentz Lie algebra has a sub-Lie bialgebra structure. With this aim,  let us recall that 
\begin{remark}
\label{re:PLcob}
All possible Lie bialgebra structures for the family $\mf g_\L$ (\ref{eq:ads_3+1}) are coboundary ones. 
\end{remark}
This fact is straightforward for $\mf{so}(3,2)$ and $\mf{so}(4,1)$ since it is well-known that all   Lie bialgebra structures for semisimple   Lie algebras   are coboundaries. 
For the  Poincar\'e algebra $\mf{iso}(3,1)$ this property was proven in~\cite{Zakrzewski1995,Zakrzewski1997,PW1996}, which   is a particular case   of the more general result stating that all Lie bialgebras for inhomogeneous pseudo-orthogonal Lie algebras $\mf{iso}(p, q)$ with $p + q \ge 3$ are coboundaries~\cite{PW1997inhomogeneous}.

Since PL groups $(G_\L, \Pi)$ are in one-to-one correspondence to   Lie bialgebras and for $\mf g_\L$ these are always coboundaries 
$(\mf g_\L, \delta_r)$ \eqref{eq:deltar},  we  thus face the problem consisting in obtaining all the possible solutions $r$ of the mCYBE \eqref{mCYBE} such that    
 the Lorentz sub-bialgebra condition (\ref{eq:subLiebialg}) is fulfilled, namely that $  \delta_r(\mathfrak{h}) \subset \mathfrak{h} \wedge \mathfrak{h}$.

 We start by noting that the 45D vector space $ \mf g_\L \wedge \mf g_\L$ admits the following ad$_{\mf h}$-invariant decomposition coming from (\ref{decom}):
\begin{equation}
\label{eq:invdecomp_lt}
  \mf g_\L \wedge   \mf g_\L= (   \mf h  \wedge   \mf h )  \oplus  ( \mf h \w \mf t  ) \oplus   (  \mf t \wedge    \mf t   ) .
\end{equation}
Hence any element $r \in  \mf g_\L\wedge \mf g_\L$ can be expressed in the form
\begin{equation}
\label{eq:r_tl}
r =  r^{\mf h \mf h} + r^{\mf h\mf t} + r^{\mf t \mf t},  \qquad  r^{\mf h \mf h} \in   \mf h \wedge    \mf h  , \qquad r^{\mf h \mf t} \in \mf h \w \mf t, \qquad   r^{\mf t \mf t} \in  \mf t \wedge    \mf t ,
\end{equation}
so that  the general element $r \in  \mf g_\L\wedge \mf g_\L$ can be written in the kinematical basis \eqref{eq:ads_3+1} as
\begin{equation}
\label{eq:general_r}
\begin{split}
r = &\sum_{a<b} r^{K_a K_b} K_a \w K_b + \sum_{a<b} r^{J_a J_b} J_a \w J_b + \sum_{a,b} r^{K_a J_b} K_a \w J_b \\
&\quad +\sum_{a,\mu} r^{K_a P_\mu} K_a \w P_\mu + \sum_{a,\mu} r^{J_a P_\mu} J_a \w P_\mu \ +\sum_{\mu<\nu} r^{P_\mu P_\nu} P_\mu \w P_\nu .
\end{split}
\end{equation}
Therefore, $r$ initially depends on 45 parameters $r^{XY}$.
From this expression, we can directly compute the cocommutator map $\delta_r : \mf g_\L \to   \mf g_\L \wedge   \mf g_\L$ using \eqref{eq:deltar}, which defines a Lie bialgebra   if and only if $r$ is a solution of the mCYBE.  

  %%%%%%%%%%%%%%%%%%%%%%%%%%%%%%%%%%%%%%%%%%%%%%%%%%%

\begin{table}[t]
{\small
\caption{\small \cite{Zakrzewski1994lorentz} Classification  of solutions of the mCYBE for the Lorentz   algebra $\mf {so}(3,1)=\spn{ \{\mathbf{K}, \mathbf{J} \}}$  in   the kinematical basis \eqref{eq:ads_3+1}.}
\label{tab:Lorentz_r}
  \begin{center}
\noindent
\begin{tabular}{l  }
\hline

\hline
\\[-0.4cm]
\\[-0.2cm] 
 Type A:\ \,$r_A=\alpha (- K_2 \wedge K_3 + J_2 \wedge J_3) +\beta (-K_2 \wedge J_3 + K_3 \wedge J_2) + \frac{\eta}{2} K_1 \wedge J_1$ \\[0.2cm]
  Type B:\ \,$r_B=\frac{\beta}{2} (K_1 \wedge K_2 + K_1 \wedge J_3 - K_3 \wedge J_1 - J_1 \wedge J_2) - \frac{\chi'}{2} (K_2 \wedge K_3 - K_2 \wedge J_2 - K_3 \wedge J_3 + J_2 \wedge J_3) $ \\[0.2cm]
  Type C:\ \,$r_C= - (\gamma + \frac{\chi'}{2}) K_2 \wedge K_3 + (\gamma - \frac{\chi'}{2}) J_2 \wedge J_3 - \gamma K_1 \wedge J_1 + \frac{\chi'}{2} (K_2 \wedge J_2 + K_3 \wedge J_3) $ \\[0.2cm]
  Type D:\ \,$r_D= \chi (K_1 \wedge K_2 + K_1 \wedge J_3)$
\\[0.25cm]
\hline

\hline
\end{tabular}
 \end{center}
}
 \end{table} 
 
 %%%%%%%%%%%%%%%%%%%%%%%%%%%%%%%%%%%%%%%%%%%%%%%%%%%

Obviously, the simplest case to be studied is    whether there exists a non-trivial Lie bialgebra structure for the family $  \mf g_\L$ such that the Lorentz subalgebra has a trivial sub-Lie bialgebra structure $\delta_r(\mathfrak{h})=0$.  By direct computation we find the following negative result:

 \begin{proposition}
 \label{prop1}
The only Poisson--Lie group $(G_\L, \Pi)$ such that $\Pi |_H = 0$ is the trivial one. 
\end{proposition}

\begin{proof}
At the Lie bialgebra level, the condition $\Pi |_H = 0$ reads $\delta_r (\mf h) = 0$. 
Hence the six cocommutators $\delta_r(K_a)$ and   $\delta_r(J_a)$ $(a=1,2,3)$  obtained from the generic $r$ (\ref{eq:general_r}) by means of \eqref{eq:deltar} must vanish.  This   amounts to solve a system of $6\times 45=270$ linear equations  with 45 unknowns $r^{XY}$. With the aid of the {\em Wolfram Mathematica} software system, it is found that their unique solution is $r = 0$.
\end{proof}

As a straightforward consequence we  find that
\begin{corollary}
The only Poisson homogeneous space  $(M_\L=G_\L/H, \pi )$ of Poisson Lorentz subgroup type such that $\Pi |_H = 0$ is the trivial one. 
\end{corollary}

Therefore   the relevant question about the existence of (non-trivial) quantum symmetries that maintain the Lorentz sector undeformed is solved.  Since these symmetries are necessarily quantizations (formal deformations) of the PL ones, we have proven that they cannot exist.

Next we investigate the existence  of PL structures for the Poincar\'e and (A)dS groups with  a non-trivial Poisson Lorentz subgroup.  This requirement  implies that $\mf  h$ has to be endowed with a  sub-Lie bialgebra structure (\ref{eq:subLiebialg}) with $\delta_r(\mathfrak{h})\ne 0$. Under this condition     the cocommutators  $\delta_r(K_a)$ and   $\delta_r(J_a)$ $(a=1,2,3)$ provided by $r$ (\ref{eq:general_r})  lead to a system of 180 linear equations whose solution is $r^{K_a P_\mu} = r^{J_a P_\mu} = r^{P_\mu P_\nu} = 0$ for all $a\in\{1,2,3\}$ and $\mu,\nu\in\{0,1,2,3\}$.  In this way, $r$  \eqref{eq:r_tl} reduces to $r = r^{\mf h \mf h} \in  \mf h\wedge  \mf h$, which means that in the kinematical basis we are using, only the terms  $r^{K_a K_b}$, $r^{K_a J_b}$ and $r^{J_a J_b}$ from (\ref{eq:general_r}) do not vanish. Moreover, these terms have to be further constrained by the mCYBE. This result can be summarized as follows:

\begin{lemma}
\label{le:rll}
Lie bialgebra structures for $\mf g_\L$ such that $\delta_r (\mf h) \subset \mf h \w \mf h$ are in one-to-one correspondence with elements $r = r^{\mf h \mf h} \in  \mf h\wedge  \mf h$ satisfying the mCYBE  (\ref{mCYBE}), i.e. 
\begin{equation}
[[r,r]]  = [[r^{\mf h \mf h}, r^{\mf h \mf h}]] \in \left( \bigwedge^3 \mf h \right)_{\mf g_\L} .
\label{lem1}
\end{equation}
\end{lemma}

We stress that the classification of the solutions of the mCYBE for the Lorentz algebra $\mf {so}(3,1)$ was given in \cite{Zakrzewski1994lorentz} (see also \cite{BLT2016unified,BLT2016unifiedaddendum} for the classification of the Euclidean and Kleinian analogues, and \cite{Kowalski2020} for their corresponding contractions leading  to (2+1)D Poincar\'e and 3D Euclidean  classical $r$-matrices). This result shows that there exists, up to Lie algebra isomorphisms, four multiparametric families of solutions which  are presented   in Table \ref{tab:Lorentz_r} in   the kinematical basis \eqref{eq:ads_3+1}. We point out that   $r_B$ and $r_D$ are solutions of the CYBE \eqref{CYBE}, while $r_A$ and $r_C$ are solutions of the mCYBE \eqref{mCYBE2} with non-vanishing algebraic Schouten bracket.

The previous Lemma together with the classification of solutions of the mCYBE for the Lorentz algebra~\cite{Zakrzewski1994lorentz} allow us to state the following main result.

\begin{proposition}
\label{prop:rPHS}
There exist three classes of Poisson homogeneous spaces $(M_\L=G_\L/H, \pi )$ for  each of the maximally symmetric spacetimes of constant curvature (Minkowski and (A)dS) such that the isotropy Lorentz subgroup  $H$ is a Poisson--Lie subgroup of $(G_\L, \Pi)$. All of them are obtained from coboundary Poisson--Lie structures on their respective isometry group $G_\L$ which are determined, up to  $\mf g_\L$-isomorphisms,  by the classical  $r$-matrices
\begin{equation}
\begin{split}
r_{\rm I} = \, &z\, (K_1 \w K_2 + K_1 \w J_3 - K_3 \w J_1 - J_1 \w J_2) \\
&\quad \ - z' \,(K_2 \w K_3 - K_2 \w J_2 - K_3 \w J_3 + J_2 \w J_3) ,\\
r_{\rm II} = \, &z \,K_1 \w J_1 ,\\
r_{\rm III} =\,  &z \,(K_1 \w K_2 + K_1 \w J_3) ,\\
\end{split}
\label{rmatrices}
\end{equation}
where $z$ and $z'$ are free parameters. These three $r$-matrices are solutions of the CYBE.
\end{proposition}

\begin{proof}
A PHS of Poisson   subgroup type  $( G_\L/H, \pi )$ can be obtained by canonical projection from a PL  group $(G_\L, \Pi)$. We already know that any PL  structure $\Pi$ on $G_\L$ is a coboundary one, so that  it is completely determined by a solution of the mCYBE whose algebraic Schouten bracket is $\mf g_\L$-invariant (\ref{mCYBE}). Therefore,   using Lemma \ref{le:rll}  and  the standard inclusion $\mf h \hookrightarrow \mf g_\L$ it follows that there exists a one-to-one correspondence between  PHS of Poisson Lorentz  subgroup type   and the   elements $r   \in  \mf h\wedge  \mf h$ in Table~\ref{tab:Lorentz_r}   fufilling  (\ref{mCYBE}) and \eqref{lem1}.
 The triangular classical $r$-matrices $r_B$ and $r_D$  already satisfy this condition, since their Schouten bracket vanishes, and lead to 
the solutions $r_{\rm I}$ and $r_{\rm III}$, respectively. On the other hand, the quasitriangular classical $r$-matrix $r_C$ verifies the relations   (\ref{mCYBE}) and  \eqref{lem1} 
  if and only if $\gamma = 0$ reducing to the particular case of $r_{\rm I}$ with $z=0$. Finally, $r_A$ fulfils (\ref{mCYBE}) and   \eqref{lem1}  whenever 
   $\alpha = \beta = 0$, thus yielding $r_{\rm II}$ which is also a solution of the CYBE~(\ref{CYBE}).
\end{proof}

Consequently, we have proven that there only exist three types of (3+1)D Poincar\'e and  (A)dS quantum deformations endowed with a non-trivial quantum Lorentz subgroup, whose underlying Lie bialgebras are  determined by $r_{\rm I}$, $r_{\rm II} $ and $r_{\rm III} $. The three classes correspond to triangular or nonstandard deformations. Types II and III are one-parametric deformations, while   type I is a two-parametric one with arbitrary deformation parameters $z$ and $z'$.

In the Poincar\'e case, the correspondence between  the  results given by Proposition~\ref{prop:rPHS}  with   the 21 multiparametric PL structures presented in~\cite{Zakrzewski1997} can be easily established.   The  two-parametric deformation of type I  for the Poincar\'e PL  group $(G_{\L=0}, \Pi)$ is just  case 5 of Table 1 in~\cite{Zakrzewski1997}. The classical $r$-matrix  $r_{\rm II}$ is provided by a   Reshetikhin twist, since both generators in $r_{\rm II}$ (\ref{rmatrices})
do commute, which can be identified with case 1  with parameters $\alpha=\tilde \alpha=0$ in the same Table.   Finally,    the deformation of type III comes from   the lower dimensional  Lorentz subalgebra $\mf {so}(2,1)$ spanned by $\{J_3,K_1,K_2\}$ and 
  $r_{\rm III}$  corresponds to case 6 with parameters $\beta_1=\beta_2=0$  in~\cite{Zakrzewski1997}.

 %%%%%%%%%%%%%%%%%%%%%%%%%%%%%%%%%%%%%%%%%%%%%%%%%%%

\subsection{The (2+1)-dimensional case} 
\label{s41}

For the sake of completeness, it is worth  comparing the (2+1)D case with the (3+1)D one characterized by Propositions \ref{prop1} and \ref{prop:rPHS}. As it will be shown in the following, the classification in $(2+1)$D is significantly different.

Let us consider the   (2+1)D counterpart $\mathfrak g^{2+1}_\L$ of the family of Lie algebras  $\mathfrak g_\L$ (\ref{eq:ads_3+1}). 
Hence  $\mathfrak g^{2+1}_\L$ is spanned by the six generators $\{P_0,P_1,P_2,K_1,K_2,J_3\}$ in such a manner that the commutation rules are given by (\ref{eq:ads_3+1})   setting the indices $a,b=1,2$ and fixing $c=3$. Thus $\mathfrak g_\L$ comprises the (2+1)D dS algebra  
$\mathfrak{so}(3,1)$ for $\Lambda > 0$, the AdS   algebra $\mathfrak{so}(2,2)$ when $\Lambda < 0$ and the Poincar\'e one  $\mathfrak{iso}(2,1)$  for $\Lambda = 0$. The corresponding  Lorentz subalgebra is   given by
\be
 \mathfrak h^{2+1}= \spn{ \{ K_1,K_2,J_3 \}}=\mathfrak{so}(2,1)\simeq\mf{sl}(2,\mathbb R).
 \label{lo3}
 \ee

We now follow a procedure similar to the above one, so that we consider  the  general element $r^{2+1}\in  \mathfrak g^{2+1}_\L \wedge \mathfrak g^{2+1}_\L$    written as
\begin{eqnarray}
&&\!\!\! \!\!\! 
r ^{2+1}= a_1 J_3 \wedge P_1 + a_2 J_3\wedge K_1 + a_3 P_0\wedge P_1 + a_4 P_0 \wedge K_1 +a_5 P_1\wedge K_1 +a_6 P_1\wedge K_2 \nonumber\\ 
 &&\qquad   +\, b_1 J_3\wedge P_2 + b_2 J_3\wedge K_2 + b_3 P_0\wedge P_2 + b_4 P_0\wedge K_2 +b_5 P_2\wedge K_2 + b_6 P_2\wedge K_1 \label{r21}\\ 
 &&\qquad  +\, c_1 J_3\wedge P_0 + c_2 K_1\wedge K_2 + c_3 P_1\wedge P_2 .
\nonumber
\end{eqnarray}
Therefore, $r ^{2+1}$ now depends on 15 parameters: $a$'s, $b$'s and $c$'s. We recall that the explicit expressions for the   general  cocommutator  $\delta_{r^{2+1} }$   (\ref{eq:deltar}) along with the equations coming from the mCYBE (\ref{mCYBE}) were presented in~\cite{JPCSadS2014} in eq.~(4.1), such that the notation corresponds here to set $J\equiv J_3$ and $\omega=-\Lambda$. In what follows we assume such results.

Firstly, if we   require that the Lorentz subalgebra $ \mathfrak h^{2+1}$ (\ref{lo3}) remains as a trivial sub-Lie bialgebra structure, that is,  $\delta_{r^{2+1} } ( \mathfrak{h}^{2+1})=0$, we obtain that the only non-vanishing parameters in $r ^{2+1}$  (\ref{r21})   are: 
$a_6=-c_1$, $ b_6=c_1$, and $ c_1$. Then, $r ^{2+1}$ reduces to the one-parameter classical $r$-matrix given by
\be
r ^{2+1}= c_1( J_3\wedge P_0 - P_1\wedge K_2 +  P_2\wedge K_1) .
\label{poin21}
\ee
 And the mCYBE provides a single equation:
 \be
 \L c_1^2=0.
\ee
Consequently, when $\L\ne 0$ the solution is the trivial one $r=0$ as in Proposition \ref{prop1}. Nevertheless, in the Poincar\'e case $\mathfrak{iso}(2,1)$ with $\L=0$, there  arises $r ^{2+1}$ (\ref{poin21}) as the single non-trivial solution that keeps the Lorentz subalgebra underformed. Remarkably enough, this solution is just Class (IV) in  the classification of equivalence classes (under automorphisms) of   $ r$-matrices for the  (2+1)D Poincar\'e algebra presented in~\cite{Stachura1998}. In fact, it is mentioned in {\em Remark 2} in~\cite{Stachura1998} that such a solution has no counterpart in higher dimensions. Furthermore $r ^{2+1}$ (\ref{poin21})  also appears in the classification of   Drinfel'd double structures  of the (2+1)D  Poincar\'e group performed in~\cite{PoinDD} as the case 0 (the `trivial' Drinfel'd double structure).
And  $r ^{2+1}$ (\ref{poin21}) has also been obtained in~\cite{Kowalski2020} (see eqs.~(4.50) and (4.52)) through a contraction approach from 
 $\mathfrak{so}(3,1)$  and $\mathfrak{so}(2,2)$.

Secondly, if we impose that $\delta_{r^{2+1} } ( \mathfrak{h}^{2+1})\subset\mathfrak{h}^{2+1} \w \mathfrak{h}^{2+1} \ne 0$ we find that the 
 non-vanishing parameters in $r ^{2+1}$  (\ref{r21}) are: $a_2$, $b_2$, $c_2$, $a_6=-c_1$, $ b_6=c_1$, and $ c_1$. Hence, 
$r ^{2+1}$  reads
\be
r ^{2+1}=a_2 J_3\wedge K_1+b_2 J_3\wedge K_2 + c_2 K_1\wedge K_2+  c_1( J_3\wedge P_0 - P_1\wedge K_2 +  P_2\wedge K_1).
\ee
And the mCYBE~\cite{JPCSadS2014}  leads  to the constraint
\be
c_2^2-a_2^2-b_2^2-4 \L c_1^2=0.
\label{constraint}
\ee
In the Poincar\'e case with $\L=0$, the solution can be reduced, via automorphisms, to $b_2=0$ and $a_2=-c_2$:
\be
r ^{2+1}=  c_2 K_1\wedge (K_2+J_3)+  c_1( J_3\wedge P_0 - P_1\wedge K_2 +  P_2\wedge K_1).
\label{sol2}
\ee
This two-parametric solution is just Class (I) in the classification~\cite{Stachura1998} and corresponds to the Drinfel'd double structure of case 1 
in~\cite{PoinDD}. For $c_1=0$ one recovers the solution $r_{\rm III}$ in Proposition~\ref{prop:rPHS}, which has also been obtained via contraction  in~\cite{Kowalski2020} (see eqs.~(4.27) and (4.41)).

Finally, from the classification of classical $r$-matrices for $\mf {so}(3,1)$ obtained  in~\cite{Zakrzewski1994lorentz} and displayed in 
 Table~\ref{tab:Lorentz_r}, it follows for $\L=+1$ that, under automorphisms, $r ^{2+1}$ (\ref{sol2}), subjected to the constraint (\ref{constraint}),  gives rise to two solutions:
 
\begin{itemize} 

\item If $b_2=0$, $a_2=-c_2$, then $\L c_1^2=c_1^2=0$. 
Hence 
\be
r ^{2+1}=  c_2 K_1\wedge (K_2+J_3),
\ee
 which is  type D in Table~\ref{tab:Lorentz_r} and 
$r_{\rm III}$ in Proposition~\ref{prop:rPHS}.

\item If $a_2=b_2=0$, then $c_2^2-4 c_1^2$ and $c_2= 2 c_1$. This leads to 
\be
r ^{2+1}=  c_1 \bigl(2 K_1\wedge K_2+  J_3\wedge P_0 - P_1\wedge K_2 +  P_2\wedge K_1\bigr),
\ee
which is the particular case of type C in  Table~\ref{tab:Lorentz_r} for $\gamma = \chi'/2\equiv c_1$ and has no (3+1)D counterpart as shown in 
Proposition~\ref{prop:rPHS}.
 
\end{itemize}

%%%%%%%%%%%%%%%%%%%%%%%%%%%%%%%%%%%%%%%%%%%%%%%%%%%

\sect{Noncommutative (A)dS and Minkowski spacetimes} 
\label{s5}

In this Section we first present explicitly the   three non-isomorphic  PHS defined by Proposition~\ref{prop:rPHS} and afterwards we perform their quantization.

The three PL structures $(G_\L, \Pi)$ provided by (\ref{rmatrices}) can be explicitly  obtained by computing the left- and right-invariant  vector fields on  $G_\L$  \eqref{eq:Gm}  and then constructing the Sklyanin bivector $\Pi$~\eqref{eq:sklyaninbiv}.
The pushforward of $\Pi$  on $G_\L$    by the canonical projection $p: G_\L \to G_\L/H$ gives rise
to the fundamental Poisson brackets for the PHS in terms of the geodesic parallel coordinates $x^\mu$. From them the corresponding expressions in terms of ambient coordinates $(\s^4,\s^\mu)$~\eqref{ambientspacecoords} can be deduced. 
 The  resulting  fundamental Poisson brackets for the three classes of PHS   are displayed in Table~\ref{table1} both in terms of local (geodesic parallel) $x^\mu$ and ambient coordinates $(\s^4,\s^\mu)$, where the latter are subjected to    the pseudosphere  constraint (\ref{pseudo}).
 
 Nevertheless, it turns out that the ambient coordinate $\s^4$  is always a central element for all three Poisson structures
 \be
   \{ \s^4 ,\s^\mu\}=0 .
   \label{s4commutes}
\ee
This fact means that all the PHS can be defined in terms of the (3+1)   spacetime coordinates $\s^\mu$ and that the pseudosphere condition  (\ref{pseudo}) can be rewritten as 
   \begin{equation}
 (\s^0)^2 - \bigl( (\s^1)^2+ (\s^2)^2+ (\s^3)^2 \bigr)=\frac{( s^4)^2 -1}{ \ka} \, .
\label{pseudo2}
\end{equation}
Hence this relation, automatically, yields a   common quadratic Casimir for the three types of PHS  in Table~\ref{table1}, which is given by
\be
\cas=(\s^0)^2 - (\s^1)^2  -(\s^2)^2 - (\s^3)^2 .
\label{cas}
\ee

 %%%%%%%%%%%%% Table 1 %%%%%%%%%%%%%%%%%%%%%%%%%%%%%%%
  
\begin{table}[hpt]
{\small
\caption{\small  The three types of (A)dS and Minkowski  Poisson homogeneous spaces with Poisson Lorentz subgroups   expressed in  geodesic parallel $x^\mu$ and ambient $\s^\mu$ spacetime coordinates  \eqref{ambientspacecoords} such that   the cosmological constant    $\L=-\eta^2$. The ambient coordinate $s^4$ always Poisson-commutes with $\s^\mu$. 
}
\label{table1}
  \begin{center}
\noindent\quad\ 
\begin{tabular}{l  }
\hline

\hline
\\[-0.2cm]
 {{\bf Type I}\quad $r_{\rm I}=  z(K_1\wedge K_2+K_1\wedge J_3-K_3\wedge J_1-J_1\wedge J_2 )- z'(K_2\wedge K_3-K_2\wedge J_2-K_3\wedge J_3+J_2\wedge J_3 )$}  \\[0.2cm]
\hline
\\[-0.2cm] 
 $\displaystyle{  \{ x^0,x^1\}=z \fun  \cos(\ro x^0) \frac{\tanh(\ro x^2)}{\ro}   }$  \\[0.35cm]
  $\displaystyle{ \{ x^0,x^2\}=-z  \,\frac{ \cos(\ro x^0) }{ \cosh(\ro x^1)} \left( \fun\,\frac{\sinh(\ro x^1)}{\ro}+\frac{\tanh^2(\ro x^3)}{\ro^2}     \right)-z'   \fun \, \frac{ \cos(\ro x^0)  \cosh(\ro x^2) }{ \cosh(\ro x^1)}\,\frac{ \tanh(\ro x^3)}{\ro  }        } $ \\[0.4cm]
     $\displaystyle{ \{ x^0,x^3\}= z\, \frac{ \cos(\ro x^0) }{ \cosh(\ro x^1)} \,\frac{ \tanh(\ro x^2)} {\ro}\, \frac{ \tanh(\ro x^3)} {\ro} +z'  \fun\, \frac{ \cos(\ro x^0) }{ \cosh(\ro x^1)}   \frac{ \sinh(\ro x^2)} {\ro} }$    \\[0.4cm]
  $\displaystyle{ \{ x^1,x^2\}=-z\left(\fun  \, \frac{ \sin(\ro x^0) }{ \ro} -\funb\, \frac{ \tanh^2(\ro x^3)}{\ro^2}  \right) + z' \fun\funb\cosh(\ro x^2)\, \frac{ \tanh(\ro x^3)}{\ro}  } $   \\[0.4cm]
    $\displaystyle{ \{ x^1,x^3\}=-z \funb\,  \frac{ \tanh(\ro x^2)}{\ro}  \,  \frac{ \tanh(\ro x^3)}{\ro}  -z' \fun\funb\,  \frac{ \sinh(\ro x^2)}{\ro}  }$\qquad  
      $ \displaystyle{ \{ x^2,x^3\}=z \fun \,  \frac{\tanh(\ro x^3)}{\ro} + z' \fun^2 \cosh(\ro x^2) }$ \\[0.40cm]
\hline
\\[-0.2cm]
  $\{ \s^0,\s^1\}=z(\s^0+\s^1) \s^2$\qquad   $\{ \s^0,\s^2\}=-z\bigl[(\s^0+\s^1) \s^1+(\s^3)^2  \bigr]-z'  (\s^0+\s^1) \s^3$ \\[0.3cm]   
$\{ \s^0,\s^3\}=z\s^2\s^3 + z'   (\s^0+\s^1) \s^2 $ \qquad \   $\{ \s^1,\s^2\}=-z  \bigl[(\s^0+\s^1)\s^0 -(\s^3)^2  \bigr] +z' (\s^0+\s^1) \s^3$       \\[0.3cm]  
   $\{ \s^1,\s^3\}=-z \s^2 \s^3 -z' (\s^0+\s^1) \s^2$\qquad
    $\{ \s^2,\s^3\}=z  (\s^0+\s^1) \s^3+z'   (\s^0+\s^1)^2$ \\[0.25cm]
  \hline
\\[-0.2cm]
 {{\bf Type II}\quad  $r_{\rm II}=  zK_1\wedge J_1$} \\[0.2cm]
\hline
\\[-0.2cm] 
$\{ x^0,x^1\}=0$\qquad    $\displaystyle{   \{ x^0,x^2\}=z \cos(\ro x^0)\, \frac{\tanh(\ro x^1)}{\ro}\,  \cosh(\ro x^2) \,\frac{ \tanh(\ro x^3)} {\ro} }$   \\[0.4cm]       $\displaystyle{  \{ x^0,x^3\}=- z \cos(\ro x^0)\, \frac{\tanh(\ro x^1)}{\ro} \,  \frac{\sinh(\ro x^2)}{\ro} }$   \qquad
  $\displaystyle{   \{ x^1,x^2\}=z\,\frac{ \sin(\ro x^0) }{ \ro} \, \cosh(\ro x^2)\,\frac{\tanh(\ro x^3)}{\ro}  } $ \\[0.4cm]
    $\displaystyle{ \{ x^1,x^3\}=-z\,\frac{ \sin(\ro x^0)}{\ro}  \,\frac{\sinh(\ro x^2)}{\ro} }$\qquad
   $\{ x^2,x^3\}=0$  
   \\[0.4cm]
\hline
\\[-0.2cm]
  $\{ \s^0,\s^1\}=0$\qquad\  $\{ \s^0,\s^2\}=z \s^1 \s^3$\qquad\ $\{ \s^0,\s^3\}=-z \s^1 \s^2$ \\[0.3cm] 
   $\{ \s^2,\s^3\}=0$ \qquad
   $\{ \s^1,\s^2\}=z\s^0\s^3$  \qquad $\{ \s^1,\s^3\}=-z\s^0\s^2$
\\[0.25cm]
 \hline
\\[-0.2cm]
 {{\bf Type III}\quad  $r_{\rm III}=  z(K_1\wedge K_2+K_1\wedge J_3 )$}  \\[0.2cm]
\hline
\\[-0.2cm] 
 $\displaystyle{  \{ x^0,x^1\}=z \fun\cos(\ro x^0) \,\frac{\tanh(\ro x^2)}{\ro}}$\qquad    $\displaystyle{ \{ x^0,x^2\}=-z  \fun \cos(\ro x^0) \,\frac{\tanh(\ro x^1)}{\ro}} $     \\[0.3cm]
  $\displaystyle{ \{ x^1,x^2\}=-z\fun\,\frac{ \sin(\ro x^0) }{ \ro}} $  \qquad 
    $\displaystyle{ \{ x^3,x^\mu\}=0 }$ \\[0.4cm]
\hline
\\[-0.2cm]
  $\{ \s^0,\s^1\}=z(\s^0+\s^1) \s^2$\qquad     $\{ \s^0,\s^2\}=-z(\s^0+\s^1) \s^1$\qquad     $\{ \s^1,\s^2\}=-z(\s^0+\s^1) \s^0$  \qquad    $\{ \s^3,\s^\mu\}=0$  \\[0.25cm] 
\hline
\\[-0.2cm]
         $\displaystyle{  \fun := \fun(x^0,x^1)=  \frac{ \sin(\ro x^0) }{ \ro}\,   \cosh(\ro x^1)+ \frac{\sinh(\ro x^1)}{\ro}    }$ \qquad 
           $\displaystyle{  \funb := \funb(x^0,x^1)=\cosh(\ro x^1) + \sin(\ro x^0) \sinh(\ro x^1)   }$
\\[0.4cm]
\hline

\hline
\\[-0.2cm]
\end{tabular}
 \end{center}
}
 \end{table} 
 
 %%%%%%%%%%%%%%%%%%%%%%%%%%%%%%%%%%%%%%%%%%%%%%

\newpage

Although the (A)dS PHS in the local coordinates $x^\mu$ present involved expressions, which do    depend explicitly on the curvature/cosmological constant parameter $\ro$ (\ref{constant}), we stress that  these are   homogeneously quadratic and $\ro$-independent for the three types of PHS in the ambient coordinates $\s^\mu$. Consequently,  the expressions  for the PHS in  terms of  $\s^\mu$ are formally the same in   the (3+1) Beltrami projective variables $q^\mu$  (\ref{beltrami}) since the latter can  simply be obtained from the former dividing by $(\s^4)^2$ (see (\ref{s4commutes})).  

Furthermore,  all  Minkowski PHS can be derived straightforwardly via the limit $\eta\to 0$. Obviously, these   are  also homogeneous quadratic  in terms of Cartesian coordinates  
and, therefore,   quite different from the well-studied  $\kappa$-Minkowski spacetime~\cite{Maslanka1993},
\be
[\hat x^0,\hat x^a]=-\frac 1{\kappa}\, \hat x^a,\qquad [\hat x^a,\hat x^b]=0,
\label{kappaM}
\ee
and   from their Lie-algebraic generalizations obtained in~\cite{Daszkiewicz2008,BorowiecPachol2009jordanian,BP2014extendedkappa} (see also references therein). In this respect, it should be noted that, to the best of our knowledge, only two noncommutative quadratic Minkowski spacetimes have been considered so far: 
\begin{itemize}

\item The (3+1)D quantum Minkowski space constructed   in~\cite{Lukierski2006}  from a twisted Poincar\'e group  that corresponds to   type II.

\item 
The (2+1)D Poisson Minkowski spacetime of   case 1 in~\cite{PoinDD}, which was  obtained from a Drinfel'd double structure  of the (2+1)D  Poincar\'e group,   which can be identified with type III. 

\end{itemize}

We also recall that Lie-algebraic deformations for the (3+1)D and (2+1)D  Minkowski spaces can also be found in~\cite{Lukierski2006} and~\cite{PoinDD}, respectively.

In the following, we  carry out the quantization  for each class of   PHS, thus  giving rise to the corresponding noncommutative spacetimes. We    construct the  noncommutative   Minkowski spaces in quantum Cartesian coordinates $\hat x^\mu$ in an explicit manner.  From them, noncommutative (A)dS spacetimes  can be straightforwardly obtained in terms of the quantum spacetime coordinates  $\hat\s^\mu$, since they are defined formally as the same Poisson quadratic algebra as the corresponding Minkowski spacetimes.

%%%%%%%%%%%%%%%%%%%%%%%%%%%%%%%%%%%%%%%%%%%%%%%%%%%

\subsection{Type I spacetimes} 
\label{s51}

This class corresponds to a two-parametric family of PHS, with arbitrary deformation parameters $z$ and $z'$, and exhibit 
  complicated expressions. As a shorthand notation, we have introduced the functions $A(x^0,x^1)$ and $B(x^0,x^1)$ given in Table~\ref{table1}.

 The vanishing cosmological constant limit $\ro\to 0$ of the  (A)dS brackets in geodesic parallel coordinates $x^\mu$ gives the Minkowski PHS. Notice that the limit of the functions $A(x^0,x^1)$ and $B(x^0,x^1)$  reads
  \be
  \lim_{\ro\to 0}A= x^0+x^1,\qquad  \lim_{\ro\to 0}B=1.
    \label{limites}
\ee
Hence the explicit Minkowski PHS is defined by the following quadratic Poisson algebra:
\bea
&&\{x^0,x^1\}=z(x^0+x^1) x^2 , \qquad 
  \{ x^0,x^3\}=zx^2x^3 + z'   (x^0+x^1) x^2  , \nonumber\\[2pt]
&& \{ x^0,x^2\}=-z\bigl[(x^0+x^1) x^1+(x^3)^2  \bigr] -z'  (x^0+x^1) x^3  , \label{space1}\\[2pt]
&& \{ x^1,x^2\}=-z  \bigl[(x^0+x^1) x^0-(x^3)^2  \bigr] +z' (x^0+x^1) x^3  , \nonumber\\[2pt]
&&\{ x^1,x^3\}=-z x^2 x^3 -z' (x^0+x^1) x^2 ,\qquad \{ x^2,x^3\}=z  (x^0+x^1) x^3+z'   (x^0+x^1)^2  ,\nonumber
\eea
which is formally identical with the (A)dS expressions given in Table~\ref{table1} in ambient variables   $ s^\mu$  (\ref{bzz})  (recall that $\s^4$ does not appear in the Poisson brackets). The  quadratic Casimir  $\cas$  (\ref{cas})  for this Poisson algebra yields
  \be
   \cas=(x^0)^2 - (x^1)^2  -(x^2)^2 - (x^3)^2 ,
   \label{cas1}
   \ee
  for any $z$ and $z'$.  
 
 Although this  two-parametric family  is quite involved, it  can be regarded   as the superposition of two particular ``subfamilies" that we proceed to quantize separately.

%%%%%%%%%%%%%%%%%%%%%%%%%%%%%%%%%%%%%%%%%%%%%%%%%%%

\subsubsection{Subfamily with $z=0$} 
\label{s511}

The Minkowski PHS (\ref{space1}) with $z=0$ can be described in a natural way   by introducing null-plane coordinates~\cite{Leutwyler78}. We consider the null-plane $\nullp_n^\tau$ orthogonal to the light-like vector $n=(1,1,0,0)$ in the classical Minkowski spacetime with   Cartesian coordinates $x=(x^0,x^a)$  and we define
\be
x^+=x^0+x^1,\qquad x^-=n\cdot x=x^0-x^1 =\tau.
\label{nullplane}
\ee
A point $x\in \nullp_n^\tau$  is labelled by the coordinates $(x^+,x^2,x^3)$, while $x^-=\tau$ will play the role of a ``time" or evolution parameter, and the chosen null-plane is associated with the boost generator $K_1$. It can be checked that  under the action of the boost transformation generated by $K_1$ (\ref{eq:Gm}),   the initial null-plane $\nullp_n^0$ $(x^-=0)$ is invariant,   the transverse coordinates $(x^2,x^3)$ remain unchanged, and $\exp(\xi \,\rho(K_1))$ maps $x^+$ into ${\rm e}^{\xi} x^+$. For our purposes, let us also recall that the generators of the Poincar\'e algebra $\mf g_0$ (\ref{eq:ads_3+1}) can be casted into three different classes  according to their commutator with  $K_1$,
\be
[K_1,X]=\gamma X,\qquad X\in \mf g_0 ,
\label{goodness}
\ee
where the parameter $\gamma$ is called  the ``goodness" of the generator $X$~\cite{Leutwyler78}; these are
\be
\begin{array}{ll}
\mbox{$\gamma=+1$:}&\quad  P_+=P_0+P_1,\quad  K_2- J_3  ,\quad K_3+J_2. \\[2pt]
\mbox{$\gamma=0$:}&\quad   K_1,\quad J_1,\quad P_2,\quad P_3 . \\[2pt]
\mbox{$\gamma=-1$:}&\quad    P_-= P_0-P_1,\quad  K_2+ J_3  ,\quad K_3-J_2.
\end{array}
\label{goodness1}
\ee
From the representation (\ref{eq:repg}) it can be checked that the seven generators with $\gamma=+1$ and $\gamma=0$ span the stability group of the initial null-plane $\nullp_n^0$, while the three remaining ones with 
$\gamma=-1$ move $\nullp_n^0$. In particular,  the transformations generated by $(K_2+ J_3)$ and $(K_3-J_2)$ rotate $\nullp_n^0$, and $\exp(\tau\rho( P_-)$
translates   $\nullp_n^0$ into $\nullp_n^{2\tau}$. Hence the latter generators, which span an Abelian subgroup, determine the dynamical evolution of  $\nullp_n^0$ with time $x^-=\tau$.

The quantization of the Minkowski PHS (\ref{space1}) with $z=0$ can   be immediately realised in terms of the noncommutative coordinates
$(\hat x^\pm=\hat x^0 \pm \hat x^1, \hat x^2,\hat x^3)$ since the   quantum coordinate $\hat x^+$ becomes a central element; namely in such null-plane coordinates this noncommutative spacetime reads
\be 
  [\hat x^-, \hat x^2]= -2 z' \hat x^+ \hat x^3  , \qquad
[\hat x^-,\hat x^3]=  2z' \hat x^+ \hat x^2 ,\qquad [\hat x^2,\hat x^3]= z'   (\hat x^+)^2 , \qquad
[\hat x^+, \, \cdot\,]=0 .  
 \label{nullspace1}
 \ee
 And  the quantum counterpart of the Casimir (\ref{cas1}) is
\be
  \hat\cas=\hat x^-  \hat x^+  -(\hat x^2)^2 - (\hat x^3)^2 .
\label{casN}
\ee
   
At this point, it is worth   calculating the cocommutator $\delta_{z'}$ coming from $r_{\rm I}$ (\ref{rmatrices}) with $z=0$ through the relation (\ref{eq:deltar})
and analyse its structure in relation with the above null-plane framework. Explicitly, $\delta_{z'}$ is given by:
 \bea
&&\delta_{z'}(P_0)= z' P_2\wedge (K_3-J_2) - z' P_3\wedge (K_2+J_3),\nonumber\\[2pt]
&&\delta_{z'}(P_1)= z' P_2\wedge (K_3-J_2) - z' P_3\wedge (K_2+J_3),\nonumber\\[2pt]
&&\delta_{z'}(P_2)= z' (P_0-P_1) \wedge (K_3-J_2)  ,\nonumber\\[2pt]
&&\delta_{z'}(P_3)= -z' (P_0-P_1)\wedge (K_2+J_3)  ,\nonumber\\[2pt]
&&\delta_{z'}(K_1)=  2 z'  (K_2+J_3)   \wedge (K_3-J_2) ,\label{coco1}\\[2pt]
&&\delta_{z'}(K_2)= -z' K_1\wedge (K_3-J_2) - z' J_1\wedge (K_2+J_3),\nonumber\\[2pt]
&&\delta_{z'}(K_3)=  z' K_1\wedge  (K_2+J_3)- z' J_1\wedge  (K_3-J_2) ,\nonumber\\[2pt]
&&\delta_{z'}(J_1)= 0 ,\nonumber\\[2pt]
&&\delta_{z'}(J_2)=  z' K_1\wedge (K_2+J_3) -z' J_1\wedge (K_3-J_2) ,\nonumber\\[2pt]
&&\delta_{z'}(J_3)= z' K_1\wedge  (K_3-J_2)+ z' J_1\wedge  (K_2+J_3) .\nonumber 
   \eea
 These expressions clearly exhibit the underlying null-plane symmetry determined by the boost generator $K_1$ according to the goodness $\gamma$ (\ref{goodness1}). It can be checked that, besides $J_1$, the three  generators with $\gamma=-1$, have a vanishing cocommutator:
 \be
 \delta_{z'}(P_-)=\delta_{z'}( K_2+ J_3 )=\delta_{z'}(K_3-J_2 )=0  .
 \ee
Moreover, in this null-plane basis   $r_{\rm I}$ (\ref{rmatrices}) with $z=0$ adopts the simple form
\be
r_{\rm I}=  z'  (K_3-J_2) \wedge  (K_2+J_3).
\label{r11}
\ee
Consequently, $r_{\rm I}$ corresponds to 
a Reshetikhin twist generated by two commuting operators, which means that  the quantum Poincar\'e algebra $U_{z'}(\mf g_0)$ can be easily constructed.
 It can also be  checked from $\delta_{z'}$ (\ref{coco1}) that 
 \be
    \delta_{z'}(\mf t)\subset \mf t \wedge \mf h,\qquad   \delta_{z'}(\mf h)\subset \mf h \wedge \mf h \ne 0 ,
    \label{xa}
\ee
as it should be. In addition, the relations (\ref{xa}) mean that    $\delta_{z'}$ does not contain any term  in $\mf t \wedge \mf t$, i.e.~$P_\mu\wedge P_\nu$, and by quantum duality  this fact implies that all the commutators $[\hat x^\mu,\hat x^\nu]$ vanish at the first-order in the quantum coordinates $\hat x^\mu$, although  higher-order terms could exist.  The latter is exactly the case here with the noncommutative Minkowski space (\ref{nullspace1}) which is defined by a homogeneous quadratic algebra. We stress that the relations  (\ref{xa}) and so the first-order brackets $[\hat x^\mu,\hat x^\nu]=0$ are satisfied by the three types I, II and III of deformations and hold for any value of $\L$.

%%%%%%%%%%%%%%%%%%%%%%%%%%%%%%%%%%%%%%%%%%%%%%%%%%%

\subsubsection{Subfamily  with $z'=0$} 
\label{s512}

We now consider the Minkowski PHS (\ref{space1}) with $z'=0$ together with the same null-plane coordinates  (\ref{nullplane}) associated with $K_1$. 
By taking the ordered monomials $(\hat x^-)^k(\hat x^+)^l(\hat x^3)^m(\hat x^2)^n$ the corresponding  noncommutative Minkowski space is found  to be
\be
\begin{array}{lll}
  [\hat x^-, \hat x^+]=2 z   \hat x^+\hat x^2, &\quad  [\hat x^-,\hat x^2]=z\hat x^-  \hat x^+  - 2 z  ( \hat x^3)^2,&\quad    
 [\hat x^-, \hat x^3]=2 z   \hat x^3\hat x^2,    \\[2pt]
 [\hat x^2, \hat x^3]=z  \hat x^+\hat x^3 ,   &\quad   [\hat x^+, \hat x^2]=- z   (\hat x^+)^2  ,&\quad    [\hat x^+, \hat x^3]=0     \, , 
\end{array}
\label{nullspace2}
\ee
 where associativity is ensured by the Jacobi identity.  Although the resulting expressions are more complicated than in the previous case,   the     quantum null-plane coordinates $(\hat x^+, \hat x^2,\hat x^3)$ again close a subalgebra and the quantum  Casimir (\ref{casN}) is the same.

 The cocommutator $\delta_{z}$ can be obtained from   $r_{\rm I}$ (\ref{rmatrices}) with $z'=0$ applying (\ref{eq:deltar}) and its structure turns out to be naturally adapted to the null-plane Poincar\'e generators   (\ref{goodness1}). However, the generators $J_1$ and $P_-$ are no longer primitive and   only  $(K_2+ J_3)$ and  $(K_3-J_2)$ have a vanishing cocommutator. Again  $r_{\rm I}$ (\ref{rmatrices}) with $z'=0$ has a simple expression in this null-plane basis:
\be
r_{\rm I}= z K_1 \wedge (K_2+J_3) + z J_1 \wedge  (K_3-J_2) .
\label{r12}
\ee
    Observe that this $r$-matrix is not defined through Reshetikhin twists and therefore is a more complicated solution of the CYBE than the previous case.

    Concerning the ``superposition" of these two subfamilies, it is clear that the two-parametric noncommutative Minkowski space   just comes out by considering  altogether (\ref{nullspace1}) and (\ref{nullspace2}) since $z$ and $z'$ are arbitrary and we have considered the same quantum coordinates for both spaces. Notice also that the order in (\ref{nullspace1})  is trivially compatible with that of  (\ref{nullspace2}). The complete $r_{\rm I}$ is just the addition of (\ref{r11}) and (\ref{r12}), which would further lead to a new  two-parametric null-plane quantum Poincar\'e  algebra $U_{z,z'}({\mf g}_0)$ determined by the two commuting generators $(K_2+ J_3)$ and  $(K_3-J_2)$. 
    
    In this respect,  it is worth stressing that $U_{z,z'}({\mf g}_0)$
           would be quite different from the known null-plane quantum Poincar\'e  algebra formerly obtained in~\cite{nullplane95}, which is determined by the generator $P_+$. In order to compare both quantum deformations,  we recall that the     noncommutative Minkowski spacetime for the latter
was obtained in~\cite{Rnullplane97} from its universal quantum $R$-matrix and in the  above null-plane basis (ruled now by the generator $K_1$ instead of   $K_3$~\cite{Rnullplane97}). This quantum spacetime has non-vanishing commutators given by
 \be
[\hat x^+,\hat x^-]=- 2 z \hat x^-,\qquad [\hat x^+,\hat x^2]=- 2 z \hat x^2 ,\qquad [\hat x^+,\hat x^3]=- 2 z \hat x^3  .
\label{nulloriginal}
\ee
Similarly to the $\kappa$-Minkowski space (\ref{kappaM}), this null-plane noncommutative Minkowski space  is also a  linear-algebraic deformation and 
for which a quadratic quantum Casimir (\ref{casN}) does not exist. For linear-algebraic generalizations of (\ref{nulloriginal}) we refer to~\cite{BP2014extendedkappa}.

Concerning the corresponding noncommutative (A)dS spacetimes, they can  be straightforwardly obtained in terms of the ambient quantum variables $(\hat s^\pm=\hat\s^0\pm \hat\s^1,\hat\s^2,\hat\s^3)$, whose commutators are given by expressions of the same form as   (\ref{nullspace1}) and  (\ref{nullspace2}).    We  also point out that  the cocommutator $\delta_{z,z'}$ obtained from $r_{\rm I}$ is $\L$-independent, and thus holds for the whole family ${\mf g}_\L$,  so that the quantum  algebra  $U_{z,z'}({\mf g}_\L)$ would  be endowed with  a common coproduct   for the (A)dS and Poincar\'e algebras with primitive generators $(K_2+ J_3)$ and  $(K_3-J_2)$. Obviously, the commutation rules for the quantum (A)dS algebras will be a $\Lambda$-deformation of the corresponding Poincar\'e ones.

%%%%%%%%%%%%%%%%%%%%%%%%%%%%%%%%%%%%%%%%%%%%%%%%%%%

\subsection{Type II spacetimes} 
\label{s52}

The limit $\ro\to 0$ of the (A)dS PHS of type II in coordinates $x^\mu$ given in Table~\ref{table1} leads to the following Minkowski PHS:
\be
\begin{array}{lll}
 \{ x^0,x^1\}=0 , &\quad   \{ x^0,x^2\}=z x^1 x^3 ,&\quad    
\{ x^0,x^3\}=-z x^1 x^2 ,    \\[2pt]
\{ x^2,x^3\}=0 ,&\quad  \{ x^1,x^2\}=zx^0x^3   ,&\quad   \{ x^1,x^3\}=-zx^0x^2    .  
\end{array}
\label{space2}
\ee
 The noncommutative Minkowski spacetime is  directly deduced   by 
 considering the ordered monomials $(\hat x^0)^k(\hat x^1)^l(\hat x^2)^m(\hat x^3)^n$ and reads   
\be
\begin{array}{lll}
[ \hat x^0, \hat x^1 ]=0 , &\quad  [\hat x^0, \hat x^2] =z \hat x^1 \hat x^3 ,&\quad    
[ \hat x^0,\hat x^3] =-z\hat  x^1 \hat x^2 ,    \\[2pt]
[ \hat x^2,\hat x^3] =0 ,&\quad [ \hat x^1,\hat x^2] =z\hat x^0 \hat x^3   ,&\quad  [\hat x^1,\hat x^3]=-z\hat x^0\hat x^2    .  
\end{array}
\label{ncspace2}
\ee
This structure can be regarded as a nonlinear generalization of two coupled 2D quadratic Euclidean algebras, for  which $\hat x^0$ and 
$\hat x^1$ play the role of rotation operators on  the 2-vector $(\hat x^2,\hat x^3)$. The behaviour of the    quantum coordinate pairs  $(\hat x^0,\hat x^1)$ and $(\hat x^2,\hat x^3)$ is  somehow inherited from  the   Reshetikhin twist defined by $r_{\rm II}$   (\ref{rmatrices}) which is given by the commuting  generators $K_1$ and $J_1$. 
It can be checked that these are the only primitive generators within the Lie bialgebra $(\mf g_\ka,\delta_{\rm II})$. 

We remark that (\ref{ncspace2}) corresponds to the   
 quadratic noncommutative Minkowski space constructed in~\cite{Lukierski2006} by following a different approach which consists of starting  with a representation of the universal quantum $R$-matrix for the twisted quantum Poincar\'e group and then applying the FRT procedure.

Note also that the  obtention of the corresponding quantum (A)dS  algebra $U_z(\mf g_\L)$ is straightforward via the Reshetikhin twist associated to $r_{\rm II}$.

%%%%%%%%%%%%%%%%%%%%%%%%%%%%%%%%%%%%%%%%%%%%%%%%%%%

\subsection{Type III spacetimes} 
\label{s53}

The third type of  PHS are written in  Table~\ref{table1} in local coordinates $x^\mu$ by using the function $A(x^0,x^1)$ as a shorthand notation. 
Before performing its quantization, let us relate these (A)dS PHS with $\ro\ne 0$  to some results already obtained in the literature.

Since $x^3$ Poisson-commutes with all the remaining coordinates,  we are dealing in fact with  a (2+1)D (A)dS PHS which can be expressed  in a more symmetric manner as
\be
 \{ x^0,x^1\}=z   \,\frac{\tanh(\ro x^2)}{\ro} \,\func,\quad\  \{ x^0,x^2\}=-z  \,\frac{\tanh(\ro x^1)}{\ro}\,\func
  ,\quad\  \{ x^1,x^2\}=-z\,\frac{ \tan(\ro x^0) }{ \ro}\,\func    ,
  \label{xb}
\ee
where we have introduced the function  $\func= \cos(\ro x^0)A(x^0,x^1)$:
\be
 \func(x^0,x^1)=\cos(\ro x^0)\left(  \frac{ \sin(\ro x^0) }{ \ro} \cosh(\ro x^1)+   \frac{\sinh(\ro x^1)}{\ro} \right) .
  \label{xc}
\ee
 In this form, the PHS (\ref{xb}) deeply resembles the (2+1)D AdS PHS  constructed  in~\cite{BHM2014plb}  by considering the AdS Lie algebra $\mf{so}(2,2)$ as a Drinfel'd double arising from the Drinfel'd--Jimbo deformation of $\mf{sl}(2,\mathbb R)$ and taking into account the results formerly obtained in~\cite{BHM2010plb}. In the notation of~\cite{BHM2014plb} the   AdS PHS reads
\be
 \{ x^0,x^1\}=-\xi   \,\frac{\tanh(\ro x^2)}{\ro} \,\fund,\qquad \{ x^0,x^2\}=\xi  \,\frac{\tanh(\ro x^1)}{\ro}\,\fund
  ,\qquad \{ x^1,x^2\}=\xi\,\frac{ \tan(\ro x^0) }{ \ro}\,\fund    .
  \label{xd}
\ee
such that $\xi$ is the deformation parameter, $\eta^2=-\L>0$, and  \be
 \fund(x^0,x^1)=\cos(\ro x^0)\left(   { \cos(\ro x^0) } \cosh(\ro x^1)+    {\sinh(\ro x^1)}  \right) .
  \label{xe}
\ee
The  apparent similarity between (\ref{xb}) and (\ref{xd}) disappears when both PHS  are expressed in ambient coordinates $\s^\mu$,  and this fact becomes more evident when the Minkowski limit $\eta\to 0$ is calculated: in that case $\func\to(x^0+x^1)$  (see (\ref{limites})), while $\fund\to 1$.
 In particular, the Minkowski PHS obtained from (\ref{xd})  is given by~\cite{BHM2014plb}
\be
 \{ x^0,x^1\}=-\xi  x^2 ,\qquad \{ x^0,x^2\}=\xi   x^1 
  ,\qquad \{ x^1,x^2\}=\xi x^0  .
  \label{xf}
\ee
Hence this structure provides a Lie-algebraic deformation of the (2+1)D Minkowski space which corresponds to the Lie algebra $\mf {so}(2,1)$. We remark that such a noncomutative Minkowski space was formerly considered  in~\cite{Matschull1998} in a (2+1) quantum gravity framework. This also appears as the case 0 in the classification of  Drinfel'd double structures for the  (2+1)D Poincar\'e   group performed  in~\cite{PoinDD}. 
The noncomutative Minkowski space $\mf {so}(2,1)\simeq \mf{sl}(2,\mathbb R)$ (\ref{xf}), linked to  the  (2+1)D Poincar\'e group, together with its Euclidean counterpart 
corresponding to    the Lie algebra $\mf {so}(3)\simeq \mf{su}(2)$,  associated with  the  3D Euclidean  group,   have   been extensively studied in (2+1)-gravity~\cite{Matschull1998,Bais1998,Bais2002,Batista2003,Majid2005,Joung2009} (see also references therein).

  In contrast to (\ref{xf}),   the limit $\eta\to 0$ of the (A)dS PHS (\ref{xb})  gives the following quadratic (2+1)D Minkowski PHS  
  \be
 \{ x^0,x^1\}=z(x^0+x^1) x^2 ,  \qquad   \{ x^0,x^2\}=-z(x^0+x^1) x^1 ,   \qquad
\{ x^1,x^2\}=-z(x^0+x^1)x^0        ,
 \label{space3}
\ee
  which is naturally adapted to   the null-plane description associated with $K_1$ that was presented in Section~\ref{s51}. 
This Poisson structure has been obtained previously within the case 1 of the classification of (2+1)D Poincar\'e Drinfel'd double structures    in~\cite{PoinDD}. 
  
If we now consider the noncommutative coordinates
$(\hat x^\pm=\hat x^0 \pm \hat x^1, \hat x^2,\hat x^3)$ and take the ordered monomials $(\hat x^-)^k(\hat x^+)^l (\hat x^2)^n$ we get the   noncommutative Minkowski spacetime  given by
 \be
 [\hat x^2, \hat x^+]=z (\hat x^+)^2 ,  \qquad  [\hat x^2,\hat x^-]=-z \hat x^- \hat x^+ ,    \qquad
[\hat  x^-,\hat x^+]=2 z \hat x^+  \hat  x^2      , \qquad  [\hat  x^3,\hat x^\mu]=0    \, .
\ \label{ncspace3}
\ee
The quantization   of the Casimir (\ref{cas1})  reads
\be
  \hat\cas=\hat x^-  \hat x^+  -(\hat x^2)^2  ,
\label{casN3}
\ee
which is the (2+1)D  version of (\ref{casN}), where obviously one could trivially add  the central term $(\hat x^3)^2$. Hence $  \hat\cas$  is related     to a (1+1)D  dS space  associated with the Lorentz algebra $\mathfrak{so}(2,1)$ in agreement with the (2+1)D character of this deformation.  The primitive generators with vanishing cocommutator  of the Lie bialgebra $(\mf g_\L,\delta_{\rm III})$, determined by $r_{\rm III}$   (\ref{rmatrices}),
are $(K_2+J_3)$ and, as expected, $P_3$.

   Finally,  we stress that  the quantum deformation determined by $r_{\rm III}$ (\ref{rmatrices}),  coming from   the lower dimensional  Lorentz subalgebra $\mf {so}(2,1)$ spanned by $\{J_3,K_1,K_2\}$ (see Section~\ref{s41}),   can be related to   the well-known nonstandard (or Jordanian) quantum deformation of $\mf {sl}(2,\mathbb R)\simeq \mf {so}(2,1)$~\cite{Ohn,Ogievetsky,BHOS1995,BH1996,Shariati}. Explicitly,  if we define   new generators in such subalgebra of $\mf g_\L$  as
\be
A=-2 K_1,\qquad A_\pm=K_2\pm J_3,
\ee
we find that they close the commutation relations of $\mf {sl}(2,\mathbb R)$ while $r_{\rm III}$   (\ref{rmatrices}) turns out to be   the so-called Jordanian twist for  $\mf{sl}(2,\mathbb R)$~\cite{Ogievetsky}, namely
\be
[A,A_\pm]=\pm 2 A_\pm,\qquad [A_+,A_-]= A,\qquad r_{\rm III}=-\frac z2 A\wedge A_+.
\ee
In this respect,  it is worth noting that the nonstandard quantum algebra $U_z(\mf {sl}(2,\mathbb R))$ has already been considered in~\cite{Herranz2000,Herranz2002} as the cornerstone 
   to obtain higher-dimensional quantum (A)dS  algebras by imposing to keep $U_z(\mf {sl}(2,\mathbb R))$ as a Hopf subalgebra, which is exactly the situation here. Such results were deduced by working in a conformal basis, and they would provide the complete quantum algebra   $U_z(\mf g_\L)$  through an appropriate change of basis from the conformal  to the kinematical one.
  %%%%%%%%%%%%%%%%%%%%%%%%%%%%%%%%%%%%%%%%%%%%%%

\sect{Noncommutative Newtonian and Carrollian  spacetimes}
\label{s6}

So far we have obtained   all the  PHS   $(G_\L/H, \pi )$ such that the Lorentz subgroup $H$ is a PL  subgroup of $(G_\L,\Pi)$ and then we have constructed the  corresponding noncommutative (A)dS and Minkowski spacetimes. In this context it is rather natural to analyse which are their non-relativistic and ultra-relativistic limits leading to noncommutative Newtonian and Carrollian spacetimes, respectively. Recall that under such limits  the Lorentz     subgroup $H$  becomes isomorphic to the 3D Euclidean group ${\rm ISO}(3)$.  In what follows we study these two limits separately.

 %%%%%%%%%%%%%%%%%%%%%%%%%%%%%%%%%%%%%%%%%%%%%%

\subsection{Non-relativistic limit}

We consider the commutation relations of the  family $\mf g_\L$ \eqref{eq:ads_3+1}, we apply the map defined by
\be
\>P\to  c^{-1}  \>P,\qquad  \>K\to  c^{-1}  \>K ,
\label{z4}
\ee
where $c$ is the speed of light, and then take the limit $c\to \infty$    obtaining the   Lie brackets given by
\bea
\begin{array}{lll}
[J_a,J_b]=\epsilon_{abc}J_c ,& \quad [J_a,P_b]=\epsilon_{abc}P_c , &\quad
[J_a,K_b]=\epsilon_{abc}K_c , \\[2pt]
\displaystyle{
  [K_a,P_0]=P_a  } , &\quad\displaystyle{[K_a,P_b]=0} ,    &\quad\displaystyle{[K_a,K_b]=0} , 
\\[2pt][P_0,P_a]=-\ka  K_a , &\quad   [P_a,P_b]=0 , &\quad[P_0,J_a]=0  .
\end{array}
\label{z5}
\eea
The three contracted Lie algebras constitute  the family of non-relativistic or Newtonian Lie algebras   $\nh_\ka$: 
   expanding Newton--Hooke (NH)   $\nh_+$ for $\ka>0$,   oscillating NH   $\nh_-$  for $\ka< 0$ and   
  Galilei   $\nh_0$ for $\ka=0$~\cite{BGH2020snyder,Bacry1968,HS1997casimir,Aldrovandi1999,HerranzSantander2008,Duval2014,Figueroa-OFarrill2018,Gomis2020a,BGGH2020kappanewtoncarroll}.

The three associated  (3+1)D Newtonian spacetimes are constructed as the coset spaces
\be
 {\rm N}_\ka/H ,\qquad H= {\rm ISO}(3)= \langle  \>K,\>J \rangle ,
\label{z6}
\ee 
where $ {\rm N}_\ka$ is the Lie group with Lie algebra  $\nh_\ka$ and $H$ is the isotropy subgroup  of rotations and (commuting) Newtonian boosts, which is isomorphic to the 3D Euclidean group ${\rm ISO}(3)$ (instead of the Lorentz one ${\rm SO}(3,1)$).  These Newtonian spacetimes are also  of constant sectional curvature equal to $\ro^2=-\ka$ but 
the time-like metric becomes degenerate providing   ``absolute-time". Hence there arises an invariant foliation under the    Newtonian group action whose leaves are defined by a constant time which is determined by a ``subsidiary" 3D non-degenerate 
  Euclidean spatial metric $g'$ restricted to each leaf of the foliation  (see, e.g.~\cite{BGH2020snyder} for explicit metric models).

By using   Lie group duality, it is found that the contraction  map (\ref{z4}) leads to the following contraction map   for     geodesic parallel and ambient coordinates
\bea
&& x^0\to x^0,\qquad x^a\to c \, x^a,\nonumber\\
&&  \s^0\to \s^0,\qquad \s^a\to c \, \s^a ,\qquad s^4\to s^4.
\label{z7}
\eea
By introducing these maps into \eqref{ambientspacecoords} and next applying the limit $c\to \infty$ we obtain the corresponding relations between Newtonian geodesic parallel and ambient coordinates:
\be 
 \s^4=\cos (\ro x^0)   , \qquad 
 \s^0=\frac {\sin( \ro x^0)}\ro   , \qquad
 \s^a = x^a  ,
 \label{z8}
\ee
such that
\be
 ( s^4)^2 +\ro^2  (\s^0)^2 =1,
\ee
to be compared to (\ref{pseudo}). In the flat Galilei space both sets of coordinates coincide: $s^\mu\equiv x^\mu$. The   main  time-like metric  and the subsidiary 3D non-degenerate 
  Euclidean spatial metric  $g'$ restricted to each leaf of the foliation turn out to be~\cite{BGH2020snyder}
\bea
&&\dd\sigma_\ka^2 =   \frac{  ({{\rm d}} \s^0)^2} {1 + \ka (\s^0)^2  } = (\dd x^0)^2,\label{z9}\\[2pt] 
&&\dd\sigma'^2  =  (\dd s^1)^2 +( \dd s^2)^2+ (\dd s^3)^2 =  (\dd x^1)^2 +( \dd x^2)^2+ (\dd x^3)^2\quad \mbox{with}\  \s^0,  x^0\  \mbox{constants}.
\nonumber
\eea

With all of these ingredients we analyse the contractions of the three types of  Lie bialgebras  determined by Proposition \ref{prop:rPHS}. We apply  the Lie bialgebra contraction (LBC) approach introduced in~\cite{BGHOS1995quasiorthogonal}: starting from a given Lie algebra contraction, here (\ref{z4}),  the   transformation of the quantum deformation parameter ($z$ and $z'$) that ensures a 
well-defined and non-trivial Lie bialgebra structure after  contraction has to be found. Such a transformation should be studied for $r$ and $\delta$ separately, since they could behave differently. In our case, the transformation of the deformation parameter is exactly the same for the three types of $r$ (\ref{rmatrices}) and their cocommutators $\delta$ (\ref{eq:deltar}).  This means that we  obtain   fundamental and coboundary LBC in such a manner that the contracted  $\delta$ coincides with the one provided by the contracted $r$ through the relation (\ref{eq:deltar}).
The resulting LBC are given by
\bea
\begin{array}{lll}
  \mbox{Type I:}&\quad z\to c^2 \,z , &\quad z'\to c^2 z' , \qquad   r_{\rm I}= zK_1\wedge K_2-z'K_2\wedge K_3. \\[2pt]
  \mbox{Type II:} &\quad z\to c \,z ,&\quad r_{\rm II}= zK_1\wedge J_1. \\[2pt]
 \mbox{Type III:}&\quad z\to c^2 \,z ,&\quad r_{\rm III}= zK_1\wedge K_2. 
\end{array}
\label{z10}
\eea
Hence the Newtonian deformation of   type III is  the particular case of   type I with $z'=0$. Therefore, all of them are just  twisted  Reshetikhin deformations (the boosts now commute) and the type II is the only one that remains invariant under contraction with respect to Proposition \ref{prop:rPHS}.

%%%%%%%%%%%%% Table 2 %%%%%%%%%%%%%%%%%%%%%%%%%%%%%%%
  
\begin{table}[t!]
{\small
\caption{\small   Noncommutative  Newtonian and Carrollian   spaces with quantum Euclidean subgroups  in terms of  quantum   $\hat x^\mu$ and   $ \hat \s^\mu$  spacetime coordinates  (\ref{z8}) and   (\ref{u5})  coming via contraction from the  noncommutative    (A)dS and Minkowski   spaces with quantum Lorentz subgroups.}
\label{table2}
  \begin{center}
\noindent 
\begin{tabular}{l  }
\hline

\hline
\\[-0.2cm]
 {Noncommutative Newtonian   spacetimes} \\[0.2cm]
\hline
\\[-0.2cm] 
  $\bullet$ Type I:\quad  $ r_{\rm I}= zK_1\wedge K_2-z'K_2\wedge K_3$ \\[0.15cm]
  $[\hat x^0,\hat x^a]=0$\qquad\!$\displaystyle{ [\hat x^1,\hat x^2]=-z\,\frac{\sin^2(\ro\hat x^0)}{\ro^2}}$ \qquad  $ [\hat x^1,\hat x^3]= 0$\qquad    $\displaystyle{[\hat x^2,\hat x^3]=z'\,\frac{\sin^2(\ro \hat x^0)}{\ro^2}}$  \\[0.3cm]
  $[\hat \s^0,\hat \s^a]=0$\qquad   $[\hat \s^1,\hat\s^2]=-z(\hat\s^0)^2$\qquad   $[\hat \s^1,\hat\s^3]=0$ \qquad   $[\hat\s^2,\hat\s^3]=z'(\hat\s^0)^2$  \\[0.25cm]
\hline
\\[-0.2cm]
  $\bullet$ Type II:\quad   $r_{\rm II}=  zK_1\wedge J_1$ \\[0.15cm]
 $[\hat  x^0,\hat x^a]=0$\qquad   $\displaystyle{ [\hat  x^1,\hat x^2]=z\,\frac{\sin(\ro \hat x^0)}{\ro}\, \hat x^3}$ \qquad   $\displaystyle{ [\hat  x^1,\hat x^3]=-z\,\frac{\sin(\ro \hat x^0)}{\ro}\, \hat x^2}$\qquad  $[\hat  x^2,\hat x^3]=0$  \\[0.3cm]
  $[\hat  \s^0,\hat \s^a]=0$\qquad\,$[\hat  \s^1,\hat \s^2]=z\hat\s^0\hat\s^3$ \qquad   $[\hat  \s^1,\hat \s^3]=-z\hat\s^0\hat\s^2$  \qquad   $[\hat  \s^2,\hat\s^3]=0$ \qquad 
 $ \hat \cas =  (\hat\s^2)^2 + (\hat\s^3)^2  $\\[0.25cm]

\hline
\\[-0.2cm]
 {Noncommutative Carrollian spacetimes} \\[0.2cm]
\hline
\\[-0.2cm]  
   $\bullet$ Type I:\quad $ r_{\rm I}= zK_1\wedge K_2-z'K_2\wedge K_3$ \\[0.2cm]
  $[\hat x^\mu,\hat x^\nu]=[\hat  \s^\mu,\hat\s^\nu]=0$ \\[0.25cm]

\hline
\\[-0.2cm]

 $\bullet$ Type II:\quad $r_{\rm II}=  zK_1\wedge J_1$ \\[0.2cm]
 $[\hat  x^0,\hat x^1]=0$\qquad\!$\displaystyle{ [\hat  x^0,\hat x^2]=z\,\frac{\tanh(\ro\hat  x^1)}{\ro}\, \cosh(\ro \hat x^2)\,   \frac{\tanh(\ro \hat x^3)}{\ro}}$   \qquad $\displaystyle{ [\hat  x^0,\hat x^3]=-z\,\frac{\tanh(\ro\hat  x^1)}{\ro}\,     \frac{\sinh(\ro \hat x^2)}{\ro}}$   \\[0.3cm]    
  $[\hat  \s^0,\hat \s^1]=0$\qquad  $[\hat  \s^0,\hat \s^2]=z\hat \s^1\hat \s^3$ \qquad  $[\hat  \s^0,\hat \s^3]=-z\hat \s^1\hat \s^2$ \qquad   $[\hat  \s^a,\hat \s^b]=[\hat  x^a,\hat x^b]=0$\qquad $  \cas= (\hat \s^2)^2 + (\hat \s^3)^2  $\\[0.25cm]
\hline

\hline
\end{tabular}
 \end{center}
}
 \end{table} 
 
 %%%%%%%%%%%%%%%%%%%%%%%%%%%%%%%%%%%%%%%%%%%%%%

 The Newtonian  PHS in local and ambient coordinates (\ref{z8}) can    be deduced either by contracting the (A)dS and Minkowski    PHS given in Table~\ref{table1} by applying the maps (\ref{z7}) and (\ref{z10}),   or by computing the left- and right-invariant vector fields and then using the Sklyanin bivector (\ref{eq:sklyaninbiv}) with   $r_{\rm I}$ and $r_{\rm II}$. The resulting   PHS  can be directly quantized since there are no ordering ambiguities. Obviously,   the contractions can also be applied to the noncommutative spaces obtained  in Sections~\ref{s51} and \ref{s52}, but observe that for the type I deformation this should be done by working  in  the basis  $\hat x^\mu$  instead of the null-plane basis, since the latter is not applicable in the non-relativistic spaces.
 Such  quantum Newtonian spaces are presented in Table~\ref{table2}.

     It is worth   stressing that  the quantum time coordinate $\hat s^0$ becomes a central element (and so $\hat x^0$ as well) reminding the ``time-absolute" character of the Newtonian spaces. This fact is consistent with  the contraction (\ref{z7}) of the Casimir (\ref{cas}) which reduces to $\cas=(\s^0)^2$.
     However the noncommutative spaces do not split into    $\hat s^0$ plus  a quantum  3-space subalgebra $\hat s^a$, since the latter involves $\hat s^0$.
     
     Let us briefly comment on the results presented in Table~\ref{table2} by writing the explicit noncommutative spaces for the Galilei case corresponding to the limit $\eta\to 0$, such that $\hat s^\mu\equiv \hat x^\mu$. The non-vanishing commutators of  the type I space  read
     \be
 [\hat x^1,\hat x^2]=-z (\hat x^0)^2 ,  \qquad [\hat x^2,\hat x^3]=z'\,(\hat x^0)^2 ,
 \label{z11}
  \ee
     which, surprisingly enough,  can be regarded as   two coupled Heisenberg--Weyl algebras sharing the quantum space coordinate $\hat x^2$, and whose   central element is determined by the square of the quantum time coordinate $\hat x^0$.  For each subfamily with  either $z$ or $z'$ equal to zero the coupling disappears leaving a   
single Heisenberg--Weyl algebra. Thus (\ref{z11}) provides  the  non-relativistic  reminiscences of the null-plane noncommutative Minkowski spaces   (\ref{nullspace1}) and (\ref{nullspace2}).

  The noncommutative Galilei space of type II   is given by
\be
[\hat  x^1,\hat x^2]=z\hat x^0\hat x^3 , \qquad    [\hat  x^1,\hat x^3]=-z\hat x^0\hat x^2 ,  \qquad    [\hat  x^2,\hat x^3]=0 ,
 \label{z12}
  \ee
which for a given eigenvalue of $\hat x^0$ can be interpreted as a quadratic 2D Euclidean algebra with  $\hat x^1$  playing the role of a rotation   on the quantum  2-space $(\hat x^2, \hat x^3)$. Hence  the quadratic  Casimir for this space is given by
\be
 \hat \cas =  (\hat x^2)^2 + (\hat x^3)^2.
 \ee
  Thus, under the non-relativistic limit,  the   noncommutative Minkowski spacetime (\ref{ncspace2}) is decoupled at the central quantum time coordinate $\hat x^0$ and  the noncommutative space  (\ref{z12}).

 Finally, regarding the   Newtonian Lie bialgebras  determined by  (\ref{z10}) note that, as a result of our approach,  both types I and II have a non-trivial Euclidean sub-Lie bialgebra $\mf h = \spn{ \{\mathbf{K}, \mathbf{J} \}}=\mathfrak{iso}(3)$, i.e.\ $\delta(\mf h)=\mf h \wedge \mf h\ne 0$.    For the two-parametric Lie bialgebra of  type I  $(\nh_\ka,\delta_{\rm I})$ the primitive generators are $\>P$ and $\>K$, while for the type II $(\nh_\ka,\delta_{\rm II})$ these are $\{K_1,J_1 ,P_1\}$. Their quantum algebras   can be straightforwardly constructed via  the corresponding  Reshetikhin twists.

%%%%%%%%%%%%%%%%%%%%%%%%%%%%%%%%%%%%%%%%%%%%%%%%%%%

\subsection{Ultra-relativistic limit}

We now introduce the    speed of light $c$  in the family   $\mf g_\L$  \eqref{eq:ads_3+1} through the map~\cite{Bacry1968,Levy-Leblond1965}
\be
P_0\to  c\, P_0,\qquad  \>K\to    c\, \>K.
\label{u1}
\ee
The  ultra-relativistic limit $c\to 0$ gives rise to  the contracted commutation relations 
\bea
\begin{array}{lll}
[J_a,J_b]=\epsilon_{abc}J_c ,& \quad [J_a,P_b]=\epsilon_{abc}P_c , &\quad
[J_a,K_b]=\epsilon_{abc}K_c , \\[2pt]
\displaystyle{
  [K_a,P_0]=0 } , &\quad\displaystyle{[K_a,P_b]=\delta_{ab} P_0} ,    &\quad\displaystyle{[K_a,K_b]= 0} , 
\\[2pt]
[P_0,P_a]=-\ka   K_a , &\quad   [P_a,P_b]=\ka  \epsilon_{abc}J_c , &\quad[P_0,J_a]=0  ,
\end{array}
\label{u2}
\eea
  corresponding to the family of three Carrollian  algebras $\ca_\ka$. This  comprises  the   para-Euclidean     $\ca_+\simeq \mathfrak{iso}(4)$ for $\ka>0$,  the
para-Poincar\'e    $\ca_- \simeq \mathfrak{iso}(3,1)$ for $\ka<0$, and the (proper) Carroll algebra $\ca_0$ for $\ka=0$~
\cite{Bacry1968,BGH2020snyder,Duval2014,Figueroa-OFarrill2018,Gomis2020a,BGGH2020kappanewtoncarroll,Levy-Leblond1965,Bergshoeff2014,KT2014carroll,Hartong2015,CGP2016carroll,BGRRV2017carrollgalilei,Tzesniewski2018,Daszkiewicz2019,Marsot2021}.

The three (3+1)D Carrollian spacetimes  are  identified with the coset spaces
\be
  {\rm C}_\ka/H ,\qquad H= {\rm ISO}(3)=\langle \>K,\>J\rangle ,
\label{u3}
\ee 
where $ {\rm C}_\ka$ is the Lie group with Lie algebra  $\ca_\ka$ and $H$ is   the isotropy subgroup isomorphic to ${\rm ISO}(3)$  spanned by rotations and (commuting) Carrollian boosts. We stress that 
the main metric for the Carrollian spacetimes is again degenerate  and has   constant  sectional  curvature equal to $-\ro^2=+\ka$~\cite{BGH2020snyder}, instead of $\ro^2=-\ka$ (as in the Lorentzian and Newtonian spacetimes),  providing a 3D ``absolute-space"; this is therefore spherical in the para-Euclidean space with $\L>0$ and hyperbolic in the para-Poincar\'e one with $\L<0$. 
There does also exist 
an invariant foliation under the Carrollian group action characterized by a  ``subsidiary" 1D  time metric  $g'$  restricted to each leaf of the foliation~\cite{BGH2020snyder}.

 In this case the  contraction  map (\ref{u1}) yields   the following transformations  for     geodesic parallel and ambient coordinates
\bea
&& x^0\to c^{-1} x^0,\qquad x^a\to  x^a,\nonumber\\
&&  \s^0\to  c^{-1}  \s^0,\qquad\, \s^a\to \s^a ,\qquad s^4\to s^4.
\label{u4}
\eea
By introducing them   into \eqref{ambientspacecoords} and   applying the limit $c\to 0$ we find the   relations between Carrollian geodesic parallel and ambient coordinates, namely
\bea
&&\s^4= \cosh ( \ro x^1) \cosh ( \ro x^2)\cosh ( \ro x^3) , \nonumber\\[2pt]
&&\s^0= x^0  \cosh  (\ro x^1) \cosh ( \ro x^2)\cosh ( \ro x^3),\nonumber\\[2pt]
&&\s^1=\frac {\sinh ( \ro x^1) }\ro   \cosh ( \ro x^2)\cosh ( \ro x^3),  \label{u5} \\
&& \s^2=\frac { \sinh ( \ro x^2)} \ro\cosh ( \ro x^3), \nonumber\\
&&  \s^3=\frac { \sinh ( \ro x^3)} \ro ,
\nonumber
\eea
fulfilling 
\begin{equation}
  ( s^4)^2   -\ro^2 \bigl( (\s^1)^2+ (\s^2)^2+ (\s^3)^2 \bigr)=1  ,\qquad \ro^2=-\ka,
   \label{u6}
\end{equation}
(see (\ref{pseudo})), showing the spherical ($\ro$ imaginary)  or hyperbolic  ($\ro$ real) nature of the underlying  3-space. Observe that for the flat Carroll space both sets of coordinates coincide $s^\mu\equiv x^\mu$.

The 3D spatial  main metric in ambient coordinates $s^a$ was introduced in~\cite{BGH2020snyder} and by means of the parametrization (\ref{u5}) this can also be written     in terms of geodesic parallel coordinates $x^a$ as
\bea
&&\dd\sigma_\ka^2 =     \frac{  \ka\left(   \s^1 \dd\s^1+\s^2 \dd\s^2+\s^3 \dd\s^3  \right)^2} {1   -\ka \bigl( (\s^1)^2+ (\s^2)^2+ (\s^3)^2 \bigr)    }+ (\dd \s^1)^2  + ( \dd \s^2)^2+ (\dd \s^3)^2
 \nonumber\\[2pt]
 &&\qquad =\cosh^2(\ro
x^2) \cosh^2(\ro x^3)(\dd x^1)^2  
+\cosh^2(\ro x^3)( \dd x^2)^2+ (\dd x^3)^2 .
 \label{u7}
\eea
The metric structure on the Carrollian spacetime is completed with    the ``subsidiary" 1D time metric  $g'$ restricted to each leaf of the foliation (the ``absolute-space") and reads
\be
\dd\sigma'^2  =  (\dd s^0)^2   =   \cosh^2  (\ro x^1) \cosh^2 ( \ro x^2)\cosh^2 ( \ro x^3)(\dd x^0)^2 
 \quad \mbox{on}\  \s^a,  x^a=\mbox{constant}.
 \label{u8}
\ee

Now we proceed to study the  ultra-relativistic limit of the solutions of the mCYBE from Proposition \ref{prop:rPHS}
by analysing their LBC. We  again obtain simultaneous fundamental and coboundary LBC~\cite{BGHOS1995quasiorthogonal}  which are defined by
\bea
\begin{array}{lll}
 \mbox{Type I:}&\quad z\to z/c^2  , &\quad z'\to z'/ c^2 , \qquad  r_{\rm I}= zK_1\wedge K_2-z'K_2\wedge K_3. \\[2pt]
  \mbox{Type II:}&\quad z\to z/c , &\quad r_{\rm II}= zK_1\wedge J_1. \\[2pt]
 \mbox{Type III:}&\quad z\to z/c^2  ,&\quad r_{\rm III}= zK_1\wedge K_2. 
\end{array}
\label{u9}
\eea
Consequently, the result is formally the same as in the Newtonian cases (\ref{z10}), so that, again,   the type III deformation is included in the type I for $z'=0$. 
The corresponding   Carrollian PHS in local and ambient coordinates (\ref{u5}) can be derived through contraction from     the Lorentzian PHS in Table~\ref{table1} by taking into account   the transformations (\ref{u4}) and (\ref{u9}). They can also be constructed    by means of the left- and right-invariant vector fields and the Sklyanin bivector (\ref{eq:sklyaninbiv}). The     resulting  PHS  have no ordering problems so that their quantization is immediate.  
The same result follows by applying   the contractions  (\ref{u4}) and (\ref{u9})  to the noncommutative spaces given  in Sections~\ref{s51} and \ref{s52}.   The  noncommutative Carrollian spaces are presented in Table~\ref{table2}.

The noncommutative Carrollian spaces of  type I are trivial ones, since all their Poisson brackets vanish.
Hence there are no Heisenberg--Weyl algebras with central quantum time coordinate $\hat s^0$, a fact that could be expected since, in the classical picture,   time is no longer absolute and its role is replaced by space.

 In the noncommutative spaces of  type II,  all the spatial coordinates commute, thus reminding the  ``absolute-space"   feature of the Carrollian spaces. Moreover,    the  quantum  spatial coordinate  $\hat s^1$ becomes a central generator (such a role was played by   the quantum  time coordinate $\hat s^0$ in the noncommutative Newtonian spaces). 
In the    proper Carroll case with $\eta\to 0$ we find that
\be
  [\hat  x^0,\hat x^2]=z\hat x^1\hat x^3 , \qquad  [\hat  x^0,\hat x^3]=-z\hat x^1\hat x^2,\qquad  [\hat  x^0,\hat x^1]=0   .
  \label{u10}
  \ee
  Hence for a fixed eigenvalue of $\hat s^1$ this noncommutative space can be seen as   a quadratic 2D Euclidean algebra with the quantum time coordinate $\hat x^0$  behaving as   a rotation   on the quantum  2-space $(\hat x^2, \hat x^3)$. Thus (\ref{u10}) has    a quadratic  Casimir   given by
\be
 \hat \cas =  (\hat x^2)^2 + (\hat x^3)^2.
 \ee
  Therefore,   under the ultra-relativistic limit,  the   noncommutative Minkowski space (\ref{ncspace2}) is now decoupled as a central quantum spatial coordinate $\hat x^1$ plus  the noncommutative Carroll space  (\ref{u10}). 
Notice that the contraction (\ref{u4}) of the Casimir (\ref{cas}) gives  $\cas=(\s^1)^2+ (  \s^2)^2 + (  \s^3)^2$ but here $  s^1$ is a central element.

As far as the  Carrollian  Lie bialgebras  provided by  (\ref{u9}) are concerned, we remark that they are again endowed with a   non-trivial Euclidean sub-Lie bialgebra $\delta(\mf h)=\mf h \wedge \mf h\ne 0$.    For the two-parametric  Lie bialgebra of  type I  $(\ca_\ka,\delta_{\rm I})$ the primitive generators are $P_0$ and $\>K$, while for the type II $(\ca_\ka,\delta_{\rm II})$ these are $\{K_1,J_1 ,P_0\}$. Again, the corresponding  quantum algebras   can be easily deduced   through         twist operators.

%%%%%%%%%%%%%%%%%%%%%%%%%%%%%%%%%%%%%%%%%%%%%%%

\sect{Concluding remarks}

The fate of Lorentz invariance in the context of deformed symmetries is a long-standing question. In this paper, we have studied this problem from a quantum group approach. Namely, we have proven that there are no possible quantum group deformations of the classical spacetime symmetries for the three maximally symmetric spacetimes of constant curvature such that the Lorentz subgroup is kept undeformed. Furthermore, we have proved that there are only three families of non-isomorphic quantum group deformations that keep the Lorentz sector as a quantum subgroup. This directly implies that there is no (non-trivial) covariant noncommutative spacetime with an undeformed Lorentz isotropy subgroup, and that there are three families of non-isomorphic noncommutative spacetimes endowed with a quantum Lorentz subgroup. 
These results differ from the (2+1)D  case  for which we have shown that only for the Poincar\'e group there exists a deformation keeping 
 the Lorentz subgroup     undeformed.

Moreover, we have explicitly constructed the Poisson version of these three noncommutative spacetimes by using both ambient and local coordinates. In terms of ambient coordinates the three Poisson bivectors give rise to PHS which are formally identical for any value of the cosmological constant, something that in our kinematical approach is indeed related to the fact that the cosmological parameter $\L$ does not appear explicitly on the Lorentz subgroup. However, by using local coordinates the spacetime curvature explicitly arises in the noncommutative structures thus distinguishing the (A)dS and Minkowski cases. We stress that   in the flat Minkowski space the three noncommutative spacetimes that we obtain are quadratic homogeneous, in striking difference with respect to the well-known $\kappa$-Minkowski spacetime. In addition, we have studied the non-relativistic and ultra-relativistic limits, obtaining the Newtonian and Carrollian noncommutative spacetimes in which, respectively, their absolute time and absolute space structure is preserved.

An interesting open problem is indeed the construction of the full Hopf algebra structures associated with the three families of quantum groups here presented, together with their corresponding quantum $\mathcal R$-matrices. In fact, 
family II   can be straightforwardly quantized through the corresponding Reshetikhin twist operator. On the other hand, a  Hopf algebra structure of the  type III was constructed in \cite{BH1996,Shariati}  for $\mf {so}(2,1)$ and for higher dimensional  (A)dS algebras in
\cite{Herranz2000,Herranz2002}. The quantum algebra corresponding to the biparametric $r$-matrix of type I  will be indeed more involved from a technical viewpoint, but this problem is worth to be faced since its associated DSR Hopf algebra model would be a novel one that would be relevant in situations when null-plane symmetry is physically relevant.

Furthermore, in order to make use of the results here presented in a General Relativistic framework, it would be essential to face the issue of general covariance for a noncommutative field theory in which any of the three novel noncommutative spacetimes here introduced is assumed as the basic local structure of the noncommutative spacetime, and whose ten-dimensional local covariance would be provided by the corresponding quantum Hopf algebras here characterized (see~\cite{Bojowald:2017kef} and references therein). In this sense, the fact that the three  spacetimes here introduced come from classical $r$-matrices which are skewsymmetric solutions to the CYBE makes it possible to construct their associated Hopf algebras from Drinfel'd twist operators $F$~\cite{Drinfeld1987icm} and, from the latter, to obtain the corresponding $\ast$-products for both the noncommutative spacetimes and the full quantum groups of local covariance~\cite{Drinfeldast}. This would open the path to the quantization-deformation of the complete group of diffeomorphisms for a General Relativistic noncommutative theory by following the approach presented in~\cite{Bojowald:2017kef}, where the hypersurface deformation algebroid (HDA) is constructed, and where the $\ast$-product on the background noncommutative spacetime plays an essential role in order to define Drinfel'd twists of spacetime diffeomorphisms and their deformed actions. In particular, the building block for this approach in the deformation of type III would be given by the so-called Jordanian twist operator, whose $\ast$-product was explicitly given in~\cite{Ohn}. Moreover, the $\ast$-products for the noncommutative Minkowski spacetimes of both the first subfamily of type I and the type II deformation should be relatively simple to construct, since both deformations are generated by a Reshetikhin twist. Also, an alternative route to an emergent Gravity theory with general  covariance in a noncommutative setting would be provided by the construction of matrix models of the type~\cite{Steinacker:2016vgf} (see also~\cite{Steinacker:2019fcb} and references therein) where the covariant Euclidean quantum spaces provided by 4D fuzzy spheres would be substituted by the Lorentzian noncommutative spacetimes with quantum group symmetry here presented. Indeed, within this approach the complete study of the representation theory of the chosen spacetime algebra and of the associated covariance quantum group would be necessary as the first step in order to analyse the feasibility of the full construction of nocommutative gauge and matter fields on this background.

%%%%%%%%%%%%%%%%%%%%%%%%%%%%%%%%%%%%%%%%%%%%%%%

\section*{Acknowledgements}

This work has been partially supported by   Agencia Estatal de Investigaci\'on (Spain)  under grant  PID2019-106802GB-I00/AEI/10.13039/501100011033,  and by Junta de Castilla y Le\'on (Spain) under grants BU229P18 and BU091G19. 
 The authors would like to acknowledge the contribution of the   European Cooperation in Science and Technology COST Action CA18108.

%%%%%%%%%%%%%%BIBLIOGRAPHY%%%%%%%%%%%%%%%

\end{document}